\documentclass[11pt]{article}
\usepackage{amsmath,amssymb,amsfonts,amsthm,latexsym}
\usepackage{graphicx,booktabs,multirow}
\usepackage{enumerate,setspace,geometry,changepage,lipsum,xr}
\usepackage[toc,title,titletoc,header]{appendix}
\usepackage{titletoc}
\usepackage{etoc}
\usepackage[bottom]{footmisc}
\usepackage[authoryear]{natbib}
\usepackage{xstring} 

\numberwithin{equation}{section}

\usepackage[pdfborder={0 0 0},colorlinks=true,linkcolor=magenta,citecolor=blue,filecolor=magenta]{hyperref}
\newcommand{\myref}[2]{\hyperref[#1]{#2}}

\usepackage{lineno,etoolbox}
\setpagewiselinenumbers

\allowdisplaybreaks

\def\qed{\rule{2mm}{2mm}}

\newtheorem{theorem}{Theorem}
\newtheorem{lemma}{Lemma}
\newtheorem{corollary}{Corollary}
\newtheorem{assumption}{Assumption}

\newtheorem{example}{Example}
\newtheorem{remark}{Remark}

\newcounter{assumptionM}
\newcounter{assumptionA}
\def\theassumptionM{M.\arabic{assumptionM}}
\def\theassumptionA{A.\arabic{assumptionA}}
\AtEndEnvironment{remark}{~\qed}%
\AtEndEnvironment{example}{~\qed}%

\begin{document}
\author{
Federico A. Bugni \\
Department of Economics\\
Northwestern University \\
\url{federico.bugni@northwestern.edu}
\and
Ivan A.\ Canay \\
Department of Economics\\
Northwestern University\\
\url{iacanay@northwestern.edu}
\and
Deborah Kim \\
Department of Economics\\
University of Warwick\\
\url{deborah.kim@warwick.ac.uk}
}

\title{Approximating an Ideal Experiment with Induced Order Statistics\footnote{We thank Matias Cattaneo for comments and discussion. All errors are our own.}}

\maketitle

\vspace{-1cm}
\begin{spacing}{1.1}
    \begin{abstract}
    Many economic questions concern the conditional distribution of an outcome at a fixed covariate value. We refer to observing i.i.d.\ draws from this distribution as an \emph{ideal experiment}, under which many statistical procedures admit exact, finite-sample guarantees. When the covariate is continuously distributed, however, this ideal experiment is unavailable. We provide a general method for approximating this ideal experiment using \emph{induced order statistics}, the outcomes attached to the nearest observed covariates. We derive explicit bounds on the approximation error in Hellinger and total variation distances. These bounds quantify how local approximation errors accumulate as the nearest-neighbor sample grows and provide a unified way to transfer rejection and coverage guarantees across statistical procedures. Under quadratic mean differentiability and local regularity of the covariate distribution, the bounds apply to both interior and boundary conditioning points and give explicit conditions on the growth of the nearest-neighbor sample. We illustrate the approach in four distinct applications:  covariate balance testing in regression discontinuity designs, conditional quantile inference, conformal prediction at a fixed covariate value, and policy-irrelevance testing under network interference.
        \end{abstract}
\end{spacing}



\thispagestyle{empty}

\newpage 

\section{Introduction}

Many economic questions concern the conditional distribution of an outcome at a fixed covariate value. Fix such a value $x_0$, and let $P$ denote the conditional distribution of an outcome $Y$ given $X=x_0$. In this paper, we refer to observing an i.i.d.\ sample from $P$ as the \emph{ideal experiment}. If this ideal experiment were available, many of these economic questions could be addressed using statistical procedures with exact, finite-sample guarantees. In a regression discontinuity design, for example, covariate balance testing would be a standard two-sample problem if we could draw independent samples of treated and untreated units at the cutoff. An i.i.d.\ sample from $P$ would similarly permit finite-sample rank-based inference on a conditional quantile. Related ideal experiments arise when constructing conformal prediction sets at a fixed covariate value and when testing policy irrelevance at a target network configuration.

If $X$ is continuously distributed, the observed sample does not provide an ideal experiment. In this case, the event $X=x_0$ has probability zero, and so the observed sample contains no observations drawn exactly from $P$. A natural approximation uses the outcomes attached to the $k$ observations whose covariates lie closest to $x_0$. These outcomes are \emph{induced order statistics} (IOS), or, equivalently for our purposes, the responses in a $k$-nearest-neighbor sample. The quality of this approximation depends on two opposing forces. Keeping $k$ fixed ensures that the selected covariates become arbitrarily close to $x_0$, but limits the information available for inference. Allowing $k$ to grow provides more observations but can also increase the discrepancy between the feasible sample and the ideal experiment.

Classical work on IOS develops their probabilistic properties, limit distributions, and representations in terms of mixtures and conditional laws; see, among others, \citet{david/galambos:74,bhattacharya:74,reiss:89,kaufmann/reiss:92} and the book-length treatment by \citet{falk/husler/reiss:2010}. 
Empirical practice often relies on asymptotic analyses that keep the dimension of the IOS vector fixed, reflecting the lack of readily available theoretical results that allow this dimension to grow. On the theoretical front, the main general result allowing $k$ to grow, \citet[Theorem~3.5.2]{falk/husler/reiss:2010}, relies on a multiplicative expansion of the joint density in a full neighborhood of $x_0$. This condition yields a fast approximation rate but restricts how the conditional support of $Y$ may vary with $X$ and rules out the possibility that $x_0$ is a boundary point. Boundary points are particularly important in regression discontinuity designs, where each side of the cutoff is analyzed separately, and arise naturally in other applications considered below.

This paper derives approximation rates for the entire IOS vector under primitive conditions that accommodate both interior and boundary conditioning points. Let $d$ denote the dimension of $X$. Let $S_n$ collect the outcomes attached to the $k$ covariate observations nearest to $x_0$, and let $S$ denote the ideal experiment, consisting of $k$ i.i.d.\ draws from $P$. We study the Hellinger and total variation distances between the two laws, $\mathcal L(S_n)$ and $\mathcal L(S)$. Hellinger distance tensorizes over product measures and describes how local approximation errors accumulate across the $k$ observations. Total variation distance controls the maximal discrepancy in event probabilities, translating the joint approximation into bounds on rejection and coverage probabilities.

Our analysis proceeds in two steps. For $r>0$, let $P_r$ denote the outcome distribution among observations whose covariates $X$ lie within distance $r$ of $x_0$. The local approximation concerns the distance between $P_r$ and $P$ as $r$ shrinks. Our high-level result shows how this discrepancy accumulates across the $k$ selected outcomes and thereby determines the joint approximation between $\mathcal L(S_n)$ and $\mathcal L(S)$. We next provide primitive conditions for the local approximation. Our main condition requires the conditional densities to be differentiable in quadratic mean at $x_0$. Combined with local regularity of the covariate distribution and its support, this condition implies that $P_r$ approaches $P$ at rate $O(r)$ in both Hellinger and total variation distances, whether $x_0$ is an interior point or a boundary point. Consequently,
\begin{align}
\mathrm{H}\big(\mathcal L(S_n),\mathcal L(S)\big)
&=O\big(k^{1/2}(k/n)^{1/d}\big)~,\notag\\
\mathrm{TV}\big(\mathcal L(S_n),\mathcal L(S)\big)
&=O\big(k^{1/2}(k/n)^{1/d}\big)~.\label{eq:intro-qmd-rate}
\end{align}
The feasible experiment therefore approximates the ideal one whenever $k=o(n^{2/(d+2)})$. For a scalar covariate, this condition becomes $k=o(n^{2/3})$. We show that the local approximation rate underlying \eqref{eq:intro-qmd-rate} is sharp if $x_0$ lies on the boundary of the support of $X$. At interior points, both local distances are $o(r)$ for every law satisfying our conditions, but there is such a model for which neither distance is $O(r^{1+\varepsilon})$ for any $\varepsilon>0$. For comparison, \citet{falk/husler/reiss:2010} obtain the less restrictive sufficient condition $k=o(n^{4/(d+4)})$ under stronger smoothness conditions at an interior conditioning point. Section~\ref{sec:fhr-benchmark} discusses the source of this difference, while the supplemental appendix provides complementary results under a Taylor/H\"older condition.

The joint total variation bound provides a general blueprint for inference with IOS. It has three steps: establish a finite-sample guarantee for a procedure in the ideal experiment, apply the same procedure to $S_n$, and use the total variation bound to quantify the resulting approximation error. The feasible guarantee and an explicit error rate then follow without a separate approximation argument for each application.

We apply this argument in four settings. First, we revisit the covariate balance test for regression discontinuity designs proposed by \citet{canay/kamat:18}. In the ideal experiment, the observations selected on the two sides of the cutoff are exchangeable under the null, and the permutation test has finite-sample size control. Writing $k$ for the total number of observations selected across the two sides, our approximation yields a rejection probability bounded by $\alpha+O(k^{3/2}/n)$. This bound quantifies the approximation error as the local sample grows and gives the sufficient restriction $k=o(n^{2/3})$ for asymptotic size control. The original Canay--Kamat rule of thumb, whose deterministic rate component is $n^{0.9}/\log n$, lies outside this guarantee. Section~\ref{sec:rdd-balance} describes a modification satisfying the restriction.

For our second application, we construct a confidence interval for a conditional quantile using two order statistics of the $k$ nearest outcomes. In the ideal experiment, their indices are calibrated using exact binomial probabilities. The finite-sample rank argument remains valid when the conditional distribution has an atom at the target quantile, in which case the interval may be conservative. Our total variation bound transfers this coverage guarantee directly to the feasible interval. The construction is closely related to the direct approach of \citet{fan/liu:16}. Their quantile-specific analysis permits a faster growth rate under smooth-interior conditions, whereas our general approximation result yields a sufficient growth condition under assumptions that accommodate boundary points.

In our third application, we construct a split-conformal prediction set for the outcome of a new observation. We study its coverage conditional on the observed covariate value, a pointwise target studied by \citet{guan:23} and \citet{min/peng/zou:26}. Conditional on an independent training sample, the usual rank argument gives finite-sample coverage in the ideal experiment. Calibrating with the nearest outcomes carries this guarantee to the feasible prediction set up to the joint total variation error. Under our main conditions, this gives coverage loss $O(k^{1/2}(k/n)^{1/d})$, including at boundary points. A separate argument exploiting conditional exchangeability sharpens the bound for this prediction rule. The argument requires no continuity of the nonconformity-score distribution; ties can only make the ideal procedure conservative.

As our final application, we study the policy-irrelevance test under network interference considered by \citet{auerbach/guo/tabord-meehan:26}. For each target, the problem can be represented as an IOS experiment after each community is reduced to the outcome attached to its closest rooted network and the corresponding minimum distance. The ideal experiment yields a finite-sample permutation argument, while the feasible procedure uses the outcomes associated with the closest configurations. Whereas the original paper establishes asymptotic validity of the test for a fixed local sample size $k$, our results establish validity with $k$ growing at an explicit rate with the number of independent communities, under a prespecified sample split and suitable local regularity conditions.

The remainder of the paper is organized as follows. Section~\ref{sec:setup} defines the ideal and feasible experiments, recalls the mixture representation of the feasible law, and presents the high-level approximation result. Section~\ref{sec:main-results} gives primitive conditions based on quadratic mean differentiability, establishes the boundary and interior results, and compares them with the classical IOS benchmark. Section~\ref{sec:application} develops the four applications, and Section~\ref{sec:conclusions} concludes. The main appendix contains the proofs and the conditional-quantile rank result; the supplemental appendix contains alternative smoothness results and their proofs.

\section{The Ideal and Feasible Experiments}\label{sec:setup}

What we call an {\it ideal experiment} consists of observing i.i.d.\ draws from the conditional distribution of $Y$ given $X=x_0$. We approximate it using \emph{induced order statistics}, the outcomes attached to the $k$ observed covariate values closest to $x_0$, and refer to observing these outcomes as the \emph{feasible experiment}. We first formalize these two experiments, then recall a mixture representation for the feasible law and state the high-level result that converts local approximation rates into rates for the entire IOS vector.

\subsection{The two experiments}\label{sec:two-experiments}

Let $X\in\mathcal X \subseteq \mathbf{R}^d$ and $Y\in\mathcal Y \subseteq \mathbf{R}^m$, and fix $x_0\in\mathcal X$. Let $\mathbf Q$ be a model class of laws of $(X,Y)$, each dominated by the product measure $d\mu := dx\otimes \nu(dy)$, where $dx$ is Lebesgue measure on $\mathbf{R}^d$ and $\nu$ is an arbitrary fixed $\sigma$-finite measure on $\mathbf{R}^m$. The supports $\mathcal X$ and $\mathcal Y$ may depend on $Q$; we suppress this dependence in the notation. The conditioning point $x_0$ remains fixed. Fix $Q\in\mathbf Q$ and denote its joint density with respect to $d\mu$ by $f(x,y)$ and the marginal density of $X$ by
\[
    g(x):=\int f(x,y)\,\nu(dy)~.
\]
We denote the conditional distribution of $Y$ given $X=x_0$ by $P$. For any $x\in\mathbf{R}^d$ with $g(x)>0$, we denote the density of the conditional distribution of $Y$ given $X=x$ by
\[
    p_x(y) := \frac{f(x,y)}{g(x)}~.
\]
This setup allows $Y$ to be continuous, discrete, or mixed, while requiring that $X$ admits a Lebesgue density.

Let $r>0$ be given. We denote the Euclidean distance of $X$ from $x_0$ by
\begin{equation}\label{eq:ranks}
    R := \|X-x_0\| ~
\end{equation}
and the open ball of radius $r$ with center $x_0$ by
\[
    B_{r}:=\{x\in \mathbf R^d:\|x-x_0\|< r\}~.
\]
We denote the conditional distribution of $Y$ given $X\in B_r$ by $P_r$ and its density by
\[
    h_r(y) := \frac{\int_{B_r} f(x,y)\,dx}{\int_{B_r} g(x)\,dx}~.
\]
Under the smoothness conditions introduced in Section~\ref{sec:main-results}, it follows that $h_r(y)$ approximates $p_{x_0}(y)$ as $r\to 0$. For notational simplicity, the dependence of $B_{r}$ and $h_r(y)$ on $x_0$ is kept implicit, since the center $x_0$ remains fixed throughout the paper.

We observe an i.i.d.\ sample of size $n$ from $Q$ and denote it by 
\begin{equation}\label{eq:obs-data}
    \{(X_i,Y_i):\,1\le i\le n\}~.
\end{equation}
We also denote the order statistics of $\{R_i:1\le i\le n\}$ in \eqref{eq:ranks} by $R_{(1)}\le \cdots\le R_{(n)}$. Since $X$ admits a Lebesgue density, ties occur with probability zero. Define the (random) permutation $\sigma_n$ of $\{1,\ldots,n\}$ by
\[
    \sigma_n(i)=j\quad\Longleftrightarrow\quad R_i=R_{(j)}~, \qquad i,j=1,\ldots,n~,
\]
so that $\sigma_n(i)$ is the rank of observation $i$.
The index set of the $k$ observations closest to $x_0$ is
\[
    K_n:=\{\,i\in\{1,\ldots,n\}:\ \sigma_n(i)\le k\,\}~.
\]
Let $\iota_n$ be the increasing rearrangement of $K_n$ (i.e., $\iota_n(1)<\cdots<\iota_n(k)$ and $\{\iota_n(1),\ldots,\iota_n(k)\}=K_n$). Finally, for $j=1,\ldots,k$, let $Y_{\iota_n(j)}$ be the response attached to the observation whose distance from $x_0$ ranks among the $k$ smallest, with the indices $\iota_n(1)<\cdots<\iota_n(k)$ listed in the original sample order. These outcomes are the \emph{induced order statistics}, listed in original sample order (equivalently, a permutation of the distance-ordered concomitants; see, e.g., \cite{david/galambos:74,bhattacharya:74,canay/kamat:18}).

Collect these observations into
\begin{equation}\label{eq:Sn}
    S_n := (S_{n,1},\ldots,S_{n,k}) := \big(Y_{\iota_n(1)},\ldots,Y_{\iota_n(k)}\big)~.
\end{equation}
By definition, $S_n$ collects the outcomes attached to the $k$ covariates closest to $x_0$. These observations can be used as an approximation to the following counterpart random vector
\begin{equation}\label{eq:S}
   S:=(S_{1},\ldots,S_{k}) := (\tilde{Y}_{1},\ldots,\tilde{Y}_{k}),\quad \text{with}\quad \tilde{Y}_{1},\ldots,\tilde{Y}_{k} \ \text{ i.i.d. with\ } P~.
\end{equation}
We refer to $S_n$ as the \emph{feasible experiment} and to $S$ as the corresponding \emph{ideal experiment}. This paper studies the approximation problem of comparing the laws of these two random vectors. Section~\ref{sec:application} gives detailed examples.

\begin{example}\label{ex:iota}
Let $m=d=1$, $n=10$, $k=4$, and suppose $x_0=0$. Consider the hypothetical sample
\[
    \{(-3,2),\ (4,-1),\ (-1,0),\ (2,5),\ (0.5,1),\ (-0.2,3),\ (7,4),\ (-5,-2),\ (1.4,6),\ (-0.7,8)\}~.
\]
In this case $R_i = |X_i|$ and so the rank map $\sigma_n$ is given by
\[
\begin{array}{c|cccccccccc}
i & 1 & 2 & 3 & 4 & 5 & 6 & 7 & 8 & 9 & 10\\
\hline
\sigma_n(i)
& 7 & 8 & 4 & 6 & 2 & 1 & 10 & 9 & 5 & 3
\end{array}~.
\]
Since $k=4$, the index set of the $k$ nearest observations is
\[
    K_n=\{ i:\sigma_n(i)\le 4\,\}=\{6,5,10,3\}~.
\]
Applying the rearrangement determined by $\iota_n$ we get $\iota_n = (3,5,6,10)$. Hence the induced order statistics are the $Y$-values attached to these indices, listed in the original sample order:
\[
    S_n = (Y_{\iota_n(1)},Y_{\iota_n(2)},Y_{\iota_n(3)},Y_{\iota_n(4)})
    = (Y_3,Y_5,Y_6,Y_{10}) = (0, 1, 3, 8)~.
\]
\end{example}

\subsection{Mixture representation and transfer}

We use $\mathcal L(\cdot)$ to denote a probability law. The mixture representation of $S_n$ makes the approximation problem transparent. 
Our assumptions below imply that $R=\Vert X-x_0\Vert$ has a continuous CDF and that $P(R<r)>0$ for every $r>0$. Under these conditions, \citet[Theorem~1]{kaufmann/reiss:92} implies that, for any $k<n$ and $r>0$, the conditional law of $S_n$ given $R_{(k+1)}=r$ is that of $k$ i.i.d.\ draws from $P_r$:
\begin{equation*}
\mathcal L\bigl(S_n\mid R_{(k+1)}=r\bigr) = P_r^k~.
\end{equation*}
If $P_r$ approximates the target conditional law $P$ for small $r>0$ and $R_{(k+1)}$ concentrates near zero, the preceding representation suggests
\begin{equation}\label{eq:prod-mixture}
\mathcal L(S_n) = \int P_r^k\,\mathcal L\bigl(R_{(k+1)}\bigr)(dr)
\approx \int P^k\,\mathcal L\bigl(R_{(k+1)}\bigr)(dr)
= P^k = \mathcal L(S)\,.
\end{equation}
In words, the $Y$-values attached to the $k$ nearest $X$-neighbors of $x_0$ behave approximately like $k$ i.i.d.\ draws from $P$. Our results below give conditions under which this approximation is valid and quantify its accuracy as $k$ grows.

This approximation has a direct implication for hypothesis testing. Let $\phi:(\mathbf R^m)^k\to[0,1]$ be any measurable function. This allows $\phi$ to represent the rejection indicator of a nonrandomized test or the rejection probability of a randomized test given the sample. Since $\phi\le 1$, the variational characterization of total variation implies
\begin{equation}\label{eq:TVbound}
    \big|E_{Q^n}[\phi(S_n)]-E_{P^k}[\phi(S)]\big|\ \le\ \mathrm{TV}\big(\mathcal L(S_n),\mathcal L(S)\big)~,
\end{equation}
where $\mathrm{TV}$ denotes the total variation distance. Moreover, the standard comparison between total variation and the Hellinger distance $\mathrm{H}$ yields
\begin{equation}\label{eq:HvsTV}
    \mathrm{H}^2\big(\mathcal L(S_n),\mathcal L(S)\big)\ \le\ \mathrm{TV}\big(\mathcal L(S_n),\mathcal L(S)\big)
\ \le\ \sqrt{2}\,\mathrm{H}\big(\mathcal L(S_n),\mathcal L(S)\big)~.
\end{equation}
Together, \eqref{eq:TVbound} and \eqref{eq:HvsTV} imply that a bound in either distance lets us carry probability guarantees from the ideal experiment to the feasible experiment, with an explicit bound on the error. This is the transfer step used in the applications below.

For estimation, let $\psi_k:(\mathbf R^m)^k\to\mathbf R ^ \ell$ estimate a parameter $\theta(P)$, and let $L$ be a loss satisfying $0\le L(\psi_k(s),\theta(P))\le \bar{L}<\infty$ uniformly in $s$ and $k$. Applying \eqref{eq:TVbound} to $L(\psi_k(\cdot),\theta(P))/\bar{L}$ gives
\[
E_{Q^n}\!\left[L\big(\psi_k(S_n),\theta(P)\big)\right]
\le
E_{P^k}\!\left[L\big(\psi_k(S),\theta(P)\big)\right]
+\bar{L}\times\mathrm{TV}\big(\mathcal L(S_n),\mathcal L(S)\big)~.
\]
The first term is the risk in the ideal $k$-sample experiment and typically decreases with $k$. For example, under squared-error loss it is of order $O(k^{-1})$ for the sample mean with bounded outcomes and for a pointwise empirical CDF. The second term is the approximation error from using the $k$ nearest observations, and our bound on this term increases with $k$ for fixed $n$. Thus, $k$ trades off statistical precision in the ideal experiment against the accuracy with which that experiment is approximated by the IOS vector.

\subsection{From local to joint approximation}

In this section, we derive rates for $\mathrm{H}\big(\mathcal L(S_n),\mathcal L(S)\big)$ and $\mathrm{TV}\big(\mathcal L(S_n),\mathcal L(S)\big)$ in terms of the local approximation errors $\mathrm{H}(P_r,P)$ and $\mathrm{TV}(P_r,P)$. We allow $x_0$ to be either an interior point or a boundary point of $\mathcal X$. The point $x_0\in\mathcal X$ is an interior point if $B_r\subset \mathcal X$ for some $r>0$, and a boundary point otherwise. This distinction is required to accommodate applications in regression discontinuity designs.

We begin with a mild regularity condition on the marginal density $g$ at $x_0$.

\begin{assumption}[Local regularity of $g$ at $x_0$]\label{ass:g-regular}
The law $Q$ satisfies the following conditions:
\vspace{-3mm}
\begin{enumerate}[(i)]
\item The marginal density $g$ of $X$ satisfies $g(x_0)>0$ and
\[
    |g(x)-g(x_0)| \le C_g \|x-x_0\|
\]
for all $x\in\mathcal X$ in a neighborhood of $x_0$ and some constant $C_g<\infty$.

\item There exist constants $C_{x_0}>0$ and $r_0>0$ such that
\[
    \mathrm{Vol}\big(B_r \cap \mathcal X\big) \ge C_{x_0} r^{d} \qquad\text{for all } 0<r<r_0~,
\]
where $\mathrm{Vol}$ denotes volume.
\end{enumerate}
\end{assumption}

Continuity of the distribution of $X$ together with positivity $g(x_0)>0$ are standard primitives in the induced order statistics literature; see, for example, \cite{reiss:89}, \cite{kaufmann/reiss:92}, and \cite{falk/husler/reiss:2010}. We strengthen these conditions slightly by imposing local Lipschitz regularity of $g$ at $x_0$ and a local thickness condition on $B_r\cap\mathcal X$, a formulation that remains compatible with $x_0$ being a boundary point. 

Although Assumption~\ref{ass:g-regular} may appear unusual at first glance, it is quite mild once one distinguishes between interior and boundary points. If $x_0$ is an interior point of $\mathcal X$, then for all sufficiently small $r$ we have $B_r\subset\mathcal X$, hence $\mathrm{Vol}(B_r \cap \mathcal X)=\mathrm{Vol}(B_r)$ is of order $r^d$, so condition~(ii) holds automatically. In this case, Assumption~\ref{ass:g-regular} reduces to condition~(i), which imposes local Lipschitz continuity of $g$ at $x_0$. If $x_0$ is a boundary point of $\mathcal X$, condition~(ii) simply requires that $\mathcal X$ occupy a nonvanishing fraction of each small ball centered at $x_0$, a property that holds, for example, when $\mathcal X$ has a Lipschitz boundary near $x_0$. Condition~(i) remains a one-sided Lipschitz condition on $g$ at $x_0$.

\begin{theorem}\label{thm:high-level}
Let Assumption \ref{ass:g-regular} hold. Suppose that
$\mathrm{H}(P_r,P)=O(r^{a_h})$ for some $a_h\ge0$ and
$\mathrm{TV}(P_r,P)=O(r^{a_{tv}})$ for some $a_{tv}\ge0$.
Then the following bounds hold uniformly for $n\in\mathbf N$ and $1\le k<n$:
\begin{enumerate}[\hspace{1cm}(a)]
    \item $\mathrm{H}\big(\mathcal L(S_n),\mathcal L(S)\big) = O\big(k^{1/2}(k/n)^{a_h/d}\big)$.
    \item $\mathrm{TV}\big(\mathcal L(S_n),\mathcal L(S)\big) = O\left(\min\Big\{k(k/n)^{a_{tv}/d}, k^{1/2}(k/n)^{a_h/d}\Big\}\right)$.
\end{enumerate}
\end{theorem}

Theorem~\ref{thm:high-level} combines the local approximation rates with the nearest-neighbor radius, whose relevant scale under Assumption~\ref{ass:g-regular} is $(k/n)^{1/d}$. Tensorization of the Hellinger affinity over the conditional product laws yields the $k^{1/2}$ factor in part~(a). For total variation, part~(b) retains the smaller of the direct product bound and the bound obtained from Hellinger distance. The latter yields the faster rate under the smoothness conditions considered in later sections. The direct bound yields a sharper rate under the alternative smoothness conditions in Appendix~\ref{supp:results-Holder}, where the local total variation distance converges faster than the Hellinger distance.

By \eqref{eq:TVbound}, the resulting total variation rate bounds the difference between probabilities under the ideal and feasible experiments. The proof of Theorem \ref{thm:high-level} adapts the argument of \citet[Theorem~3.5.2]{falk/husler/reiss:2010} to accommodate general local exponents and to obtain the direct total variation bound.

\section{Primitive Conditions and Boundary-Valid Rates}\label{sec:main-results}

The high-level result in Theorem~\ref{thm:high-level} reduces the joint approximation problem to the local behavior of $P_r$ around $P$. In this section, we establish the required local rates under quadratic mean differentiability, study their sharpness at boundary and interior points, and compare them with the smooth-interior benchmark that is currently available in the literature.

\subsection{Quadratic Mean Differentiability}\label{sec:qmd}

We develop local approximation rates for the conditional densities $p_x$ at $x_0$ under quadratic mean differentiability (QMD). We show that QMD yields linear local approximation rate of $P_r$ to $P$ in both Hellinger and total variation distances. We also study sharpness and the limits of possible improvements over the linear rate, distinguishing what can and cannot be strengthened at boundary points versus interior points.

\begin{assumption}[QMD at $x_0$]\label{ass:QMD}
The conditional densities $p_x$ of $Q$, defined in Section~\ref{sec:two-experiments}, are such that there exists a measurable function $\dot\ell_{x_0}:\mathbf R^m\to\mathbf R^d$ satisfying
\[
    \int \|\dot\ell_{x_0}(y)\|^2 p_{x_0}(y) \nu(dy)<\infty~,
\]
and, as $t\to 0$ through values satisfying $x_0+t\in\mathcal X$,
\[
    \int \Big[ \sqrt{p_{x_0+t}(y)}-\sqrt{p_{x_0}(y)}-\tfrac12\big(t^\top\dot\ell_{x_0}(y)\big)\sqrt{p_{x_0}(y)} \Big]^2 \nu(dy)
= o(\|t\|^2)~.
\]
\end{assumption}

QMD is a central smoothness condition in asymptotic statistics. It underlies the local asymptotic normality (LAN) framework and plays a key role in the asymptotic theory of likelihood-based procedures (see, e.g., \cite{lecam:86}, \cite{vandervaart:00}), with $\dot\ell_{x_0}$ serving as the score function. Beyond likelihood theory, QMD also supports diffusion/LAN approximations in statistical decision problems, such as optimal design and inference for adaptive experiments and bandits, by justifying Gaussian limit experiments and the associated risk analyses; see \citet{adusumilli:bandits}.

QMD imposes weaker restrictions on the local behavior of the conditional law than the smooth-interior conditions proposed by \citet{falk/husler/reiss:2010}. We make this comparison precise in Section~\ref{sec:fhr-benchmark}, where we introduce the latter conditions.
 The next theorem establishes the local approximation rates under QMD, gives a boundary example attaining the linear rate, and gives an interior example admitting no positive polynomial improvement. For the existence claims below, the dominating measure $\nu$ may be chosen as part of the construction.

\begin{theorem}\label{thm:H-QMD}
    Suppose Assumptions~\ref{ass:g-regular}--\ref{ass:QMD} hold. Then,
    \begin{equation}\label{eq:H-TV-QMD}
        \mathrm{H}(P_r,P)=O(r)
        \quad\text{and}\quad
        \mathrm{TV}(P_r,P)=O(r)~.    
    \end{equation}
    Moreover:
    \begin{enumerate}[(i)]
    \item \emph{Boundary sharpness.} There exists a law $Q\in\mathbf Q$ satisfying Assumptions~\ref{ass:g-regular}--\ref{ass:QMD} such that $x_0$ lies on the boundary of $\mathcal X$, and constants $C>0$ and $\bar r>0$ such that
    \[
    \mathrm{H}(P_r,P)\ge C r
    \quad\text{and}\quad 
    \mathrm{TV}(P_r,P)\ge C r
    \quad\text{for all } r\in(0,\bar r)~.
    \]
    
  \item \emph{Interior refinement and its limitation.} For every $Q\in\mathbf Q$ satisfying Assumptions~\ref{ass:g-regular}--\ref{ass:QMD} such that $x_0$ lies in the interior of $\mathcal X$,
    \[
    \mathrm H(P_r,P)=o(r)
    \quad\text{and}\quad
    \mathrm{TV}(P_r,P)=o(r)~.
    \]
    However, for every $\tilde r>0$, there exists a law $Q\in\mathbf Q$ satisfying these assumptions, with $B_{\tilde r}\subseteq\mathcal X$, such that, for every $\varepsilon>0$,
    \begin{equation}\label{eq:no-polynomial}
    \frac{\mathrm H(P_r,P)}{r^{1+\varepsilon}}\longrightarrow\infty
    \qquad\text{and}\qquad
    \frac{\mathrm{TV}(P_r,P)}{r^{1+\varepsilon}}\longrightarrow\infty
    \qquad\text{as }r\to 0~.
    \end{equation}
    \end{enumerate}
\end{theorem}

Theorem~\ref{thm:H-QMD} gives the local upper rate $O(r)$ in both distances at interior and boundary points. Part~(i) shows that this rate is sharp at boundary points by providing a boundary example that attains it. Part~(ii) shows that, at interior points, both distances are $o(r)$ for every fixed law satisfying the assumptions. Nevertheless, the $O(r)$ rates are {\it nearly sharp} in the sense that there exists a single interior QMD model for which neither distance is $O(r^{1+\varepsilon})$ for any $\varepsilon>0$.

Theorem~\ref{thm:H-QMD} implies that Theorem~\ref{thm:high-level} applies with $a_h=a_{tv}=1$, yielding
\begin{equation*}
   \mathrm{H}\big(\mathcal L(S_n),\mathcal L(S)\big) = O\big(k^{1/2}(k/n)^{1/d}\big)~,
   \qquad
   \mathrm{TV}\big(\mathcal L(S_n),\mathcal L(S)\big) = O\big(k^{1/2}(k/n)^{1/d}\big)~.
\end{equation*}
Since $a_h=a_{tv}=1$, the relevant bound for $\mathrm{TV}\big(\mathcal L(S_n),\mathcal L(S)\big)$ is provided by the inequality $\mathrm{TV} \le \sqrt{2}\mathrm{H}$, and both distances share the same joint rate. This rate directly determines how fast $k$ may grow relative to $n$ while still ensuring convergence of $\mathcal L(S_n)$ to $\mathcal L(S)$:
\[
    k^{1/2}(k/n)^{1/d} \to 0      \qquad\Longleftrightarrow\qquad   k = o\big(n^{2/(2+d)}\big)~.
\]
In particular, for $d=1$ this becomes $k = o(n^{2/3})$. This threshold coincides with the growth conditions derived, through different arguments, in \cite{bugni/canay:21} and \cite{bugni/canay/kim:2025} for the specific uses of induced order statistics in those papers. 

\begin{remark}[Joint approximation at interior points]\label{rem:interior-joint}
Suppose Assumptions~\ref{ass:g-regular}--\ref{ass:QMD} hold and $x_0$ lies in the interior of $\mathcal X$. For each fixed law $Q$ and every sequence $1\le k<n$ with $k/n\to0$,
\begin{align*}
    \mathrm H\big(\mathcal L(S_n),\mathcal L(S)\big)
    &=o\big(k^{1/2}(k/n)^{1/d}\big)~,\\
    \mathrm{TV}\big(\mathcal L(S_n),\mathcal L(S)\big)
    &=o\big(k^{1/2}(k/n)^{1/d}\big)~.
\end{align*}
Consequently, $k=O(n^{2/(d+2)})$ suffices for joint convergence at an interior point. This refinement is pointwise in $Q$.
\end{remark}

Theorem~\ref{thm:H-QMD} gives our default local approximation rate. Appendix~\ref{supp:results-Holder} provides complementary results under the Taylor/H\"older remainder condition in Assumption~\ref{ass:smoothness-holder}.  This condition and the QMD condition in Assumption~\ref{ass:QMD} are nonnested. Assumption~\ref{ass:smoothness-holder} indexes smoothness by the Taylor order $\kappa_{\rm s}\in\mathbf N_0$ and the H\"older exponent $\kappa_{\rm r}\in(0,1]$. Theorem~\ref{thm:Holder-TV-H} gives exponents $(a_h,a_{tv})$ that depend on their sum and on whether $x_0$ is an interior point or not. 
At interior points, first-order terms cancel when integrated over a symmetric neighborhood. At boundary points, the support truncates this neighborhood, so cancellation may fail and rates can be slower. For example, in the H\"older-continuity case $\kappa_{\rm s}=0$, $\mathrm{H}(P_r,P)=O(r^{\kappa_{\rm r}/2})$ and $\mathrm{TV}(P_r,P)=O(r^{\kappa_{\rm r}})$. As $\kappa_{\rm r}\downarrow0$, the admissible growth rate for $k$ becomes correspondingly more restrictive; no common positive polynomial growth rate is available uniformly over $\kappa_{\rm r}\in(0,1]$.

\subsection{Comparison with the Standard Smooth-Interior Benchmark}\label{sec:fhr-benchmark}

The monograph \cite{falk/husler/reiss:2010} provides an influential benchmark for the convergence of induced order statistics in Hellinger distance. Their result allows $k$ to grow with $n$ under the following condition.

\begin{assumption}[FHR assumption]\label{ass:book}
There exists $t_0>0$ such that, uniformly for $t\in\mathbf{R}^d$ with $\|t\|\le t_0$ and all $y\in\mathbf{R}^m$,
\begin{equation}\label{eq:book-expansion}
    f(x_0+t,y)=f(x_0,y)\Big\{1+t^\top\zeta_1(y)+O\big(\|t\|^2\zeta_2(y)\big)\Big\}~,
\end{equation}
where $\zeta_1:\mathbf{R}^m\to\mathbf{R}^d$, $\zeta_2:\mathbf{R}^m\to[0,\infty)$, and
\begin{equation}\label{eq:book-L2integrability}
    \int \big(\|\zeta_1(y)\|^2+\zeta_2(y)^2\big) f(x_0,y)\nu(dy)<\infty~.
\end{equation}
\end{assumption}

The following result uses the preceding assumption to bound the joint approximation error between the feasible and ideal experiments.

\begin{theorem}[Falk--H\"usler--Reiss, Theorem 3.5.2]\label{thm:FHR-352}
Suppose $g(x_0)>0$ and Assumption~\ref{ass:book} holds. Then, uniformly over $n\in\mathbf N$ and $k\in\{1,\ldots,n\}$,
\[
\mathrm{H}\big(\mathcal L(S_n),\mathcal L(S)\big) = O\big(k^{1/2}(k/n)^{2/d}\big)~.
\]
\end{theorem}

We now briefly explain the idea behind the proof of Theorem \ref{thm:FHR-352}. Under Assumption~\ref{ass:book} and $g(x_0)>0$, Lemma~\ref{lem:FHR-Lp} establishes $\mathrm{H}(P_r,P)=O(r^2)$ and $\mathrm{TV}(P_r,P)=O(r^2)$. Part~(a) of Theorem~\ref{thm:high-level} therefore reproduces the Hellinger rate in Theorem~\ref{thm:FHR-352}. Part~(b) also yields
\[
\mathrm{TV}\big(\mathcal L(S_n),\mathcal L(S)\big)=O\big(k^{1/2}(k/n)^{2/d}\big)~,
\]
so the same rate transfers finite-sample probability guarantees from the ideal experiment to the feasible one under the FHR conditions. The corresponding sufficient growth condition is $k=o(n^{4/(d+4)})$, compared with $k=o(n^{2/(d+2)})$ under Assumptions~\ref{ass:g-regular}--\ref{ass:QMD}.

The previous paragraph shows that the standard smooth-interior conditions, given by Assumption~\ref{ass:book} and $g(x_0)>0$, permit faster growth of $k$ than Assumptions~\ref{ass:g-regular}--\ref{ass:QMD}. This improvement relies on stronger assumptions. Formally, we show in Lemma~\ref{lem:book-implies-QMD} in the supplemental appendix that Assumption~\ref{ass:book}, together with $g(x_0)>0$, implies Assumptions~\ref{ass:g-regular} and \ref{ass:QMD}, with $x_0\in{\rm int}(\mathcal X)$, but the converse does not hold. The next paragraphs explain the key features of Assumption~\ref{ass:book} that contribute to the faster FHR rate.

First, Assumption~\ref{ass:book} requires the expansion in \eqref{eq:book-expansion} to hold for all perturbations in a full neighborhood of $x_0$. For any $(x_0,y)\in\mathcal X\times\mathbf R^m$ with $f(x_0,y)>0$, this expansion implies $f(x_0+t,y)>0$ for all sufficiently small perturbations $t\in\mathbf R^d$. Thus, $x_0$ must be an interior point of $\mathcal X$. Otherwise, arbitrarily small perturbations could lie outside the support, where $f(x_0+t,y)=0$. At an interior point, integration over the symmetric ball $B_r$ cancels the leading first-order contribution. At a boundary point, the support truncates the ball, so this cancellation need not occur.

Second, \eqref{eq:book-expansion} imposes a quadratic bound on the relative remainder, $O(\|t\|^2\zeta_2(y))$, where $\zeta_2$ is square-integrable under $P$. Together with the cancellation of the leading first-order contribution, this bound yields an approximation error of order $r^2$. By contrast, QMD in Assumption~\ref{ass:QMD} requires only that the $L^2(\nu)$ remainder in the first-order expansion of $\sqrt{p_{x_0+t}}$ be $o(\|t\|)$, without specifying a polynomial rate. This is also the mechanism behind Theorem~\ref{thm:H-QMD}(ii): at an interior point, symmetry reduces the averaged first-order contribution to $O(r^2)$, but QMD supplies no polynomial rate for the remaining term. The interior example in Theorem~\ref{thm:H-QMD}(ii) shows that no positive polynomial improvement over $O(r)$ may be attainable, even for a single model.

The multiplicative form in \eqref{eq:book-expansion} has a further implication for the zero-density region
\[
\mathcal N:=\{\,y\in\mathcal Y:f(x_0,y)=0\,\}
          =\{\,y\in\mathcal Y:p_{x_0}(y)=0\,\}~.
\]
Because the remainder in \eqref{eq:book-expansion} is multiplicative in $f(x_0,y)$, the assumption implies that $f(x_0+t,y)=0$ for every $y\in\mathcal N$ and all sufficiently small $t$. It therefore excludes models in which $y$-values having zero density at $x_0$ acquire positive density as $x$ moves locally away from $x_0$. QMD in Assumption~\ref{ass:QMD} permits such local changes, requiring only that
\[
\int_{\mathcal N}p_{x_0+t}(y)\,\nu(dy)=o(\|t\|^2)~.
\]
Thus, QMD controls the probability entering $\mathcal N$ in an integrated sense, whereas Assumption~\ref{ass:book} rules it out pointwise. We examine the role of $\mathcal N$ under alternative smoothness conditions further in Appendix~\ref{sec:auxiliary-results-holder}.

In summary, the faster FHR rate reflects both symmetric interior geometry and a quadratic bound on the relative remainder. The FHR conditions also restrict local changes in the support of the conditional law. Our QMD framework does not impose these restrictions and accommodates boundary points, as arise in regression discontinuity designs.

\section{Applications}\label{sec:application}

Fix $\alpha\in(0,1)$. A finite-sample guarantee in the ideal experiment continues to hold when the same procedure is applied to the nearest outcomes, up to the total variation distance between the two experiments. The four applications below use this transfer for rejection and coverage probabilities. Section~\ref{sec:conditional-conformal} also illustrates how additional structure can sharpen the generic coverage bound.

\subsection{Covariate Balance Testing in Regression Discontinuity Designs}
\label{sec:rdd-balance}

In a sharp regression discontinuity design, treatment is assigned when a real-valued running variable $X$ crosses a known cutoff, normalized to zero. A standard specification check asks whether the distribution of a baseline covariate $Y$ is continuous at the cutoff. Write
\[
P_-:=\mathcal L(Y\mid X=0^-),
\qquad
P_+:=\mathcal L(Y\mid X=0^+)~.
\]
Here $P_-$ and $P_+$ denote the respective one-sided limits of the conditional law of $Y$ at the cutoff. The null hypothesis studied by \citet{canay/kamat:18} is
\begin{equation}\label{eq:rdd-null}
H_0:P_-=P_+~.
\end{equation}
This is the distributional analog of the continuity-of-means checks commonly reported in empirical applications of regression discontinuity designs.

In the ideal experiment, covariate balance is a standard two-sample problem. As before, $S$ and $S_n$ denote the pooled ideal and feasible samples, respectively. Let $S_-=(Y_1^-,\ldots,Y_q^-)$ and $S_+=(Y_1^+,\ldots,Y_q^+)$ be independent, with laws $P_-^q$ and $P_+^q$, respectively, and set $S:=(S_-,S_+)$. Under $H_0$, the components of $S$ are $2q$ i.i.d. draws from a common distribution, and hence $S$ is exchangeable. Consider the Cram\'er--von Mises permutation test proposed by \citet{canay/kamat:18}, which compares the empirical distributions of the two blocks. Its critical values are computed by permuting the pooled sample. Randomization at the critical value gives rejection probability $\alpha$ in the ideal experiment; without this randomization, the rejection probability is at most $\alpha$. Let $\phi$ denote either version of the test, so that $E[\phi(S)]\le\alpha$.

When $X$ is continuously distributed, the two ideal samples are unavailable. \citet{canay/kamat:18} approximate them using the outcomes attached to the $q$ observations closest to the cutoff on each side. Let $\iota_n^-(1),\ldots,\iota_n^-(q)$ denote the indices of the $q$ observations with $X_i<0$ closest to zero, listed in their original sample order. Define $\iota_n^+(1),\ldots,\iota_n^+(q)$ analogously using the observations with $X_i>0$. The resulting feasible experiment is
\[
S_n:=\big(Y_{\iota_n^-(1)},\ldots,Y_{\iota_n^-(q)},
       Y_{\iota_n^+(1)},\ldots,Y_{\iota_n^+(q)}\big)~.
\]
Applying the same permutation rule to $S_n$ gives its feasible test, $\phi(S_n)$. Thus, the remaining question is how closely the law of $S_n$ approximates the law of $S$.

To answer this question, let $Q_-$ be the law of $(-X,Y)$ conditional on $X<0$ and let $Q_+$ be the law of $(X,Y)$ conditional on $X>0$. For $u>0$, let $P_{-,u}$ and $P_{+,u}$ denote versions of the conditional laws of $Y$ under $Q_-$ and $Q_+$, respectively, at covariate value $u$. Each block is a one-dimensional IOS experiment with conditioning point zero, a boundary point of its covariate support.

\begin{corollary}\label{cor:rdd-balance}
Suppose $Q\{X<0\}>0$ and $Q\{X>0\}>0$. Suppose also that $P_{-,u}$ and $P_{+,u}$ extend to zero with $P_{-,0}=P_-$ and $P_{+,0}=P_+$, and that Assumptions~\ref{ass:g-regular}--\ref{ass:QMD} hold at $x_0=0$ for $Q_-$ and $Q_+$ with these extensions. If $q=o(n)$, Theorems~\ref{thm:high-level}--\ref{thm:H-QMD} imply
\[
\mathrm{TV}\!\left(\mathcal L(S_n),\mathcal L(S)\right)
=O\big(q^{3/2}/n\big)~.
\]
Under $H_0$, it follows that
\begin{equation}\label{eq:rdd-size-bound}
E_{Q^n}[\phi(S_n)]
\le \alpha+O\big(q^{3/2}/n\big)~.
\end{equation}
Consequently, asymptotic size control holds if
\begin{equation}\label{eq:q-n23}
q=o\big(n^{2/3}\big)~.
\end{equation}
\end{corollary}

The proof, given in Appendix~\ref{app:application-proofs}, applies the IOS approximation separately to the two one-sided samples and combines the bounds using the product inequality for total variation. The boundary-valid rate in Section~\ref{sec:main-results} is needed because zero is a boundary point for each sample.

\begin{remark}[Choice of $q$]\label{rem:rot}
The rule of thumb proposed by \citet{canay/kamat:18} takes the form $q_{\rm rot}=C_{\rm ck}n^{0.9}/\log n$, where $C_{\rm ck}$ is the adjustment derived there under a bivariate normal model. The factor $n^{0.9}/\log n$ does not satisfy \eqref{eq:q-n23}, so Corollary~\ref{cor:rdd-balance} does not cover the original rule. Since \eqref{eq:q-n23} restricts the growth rate rather than the multiplicative constant, we suggest retaining this adjustment and modifying the deterministic rate component as follows:
\[
q_{\rm rot}=C_{\rm ck}n^\gamma,\qquad 0<\gamma<2/3~.
\]
\end{remark}

\subsection{Conditional Quantile Inference}\label{sec:conditional-quantile}

Suppose that $Y$ is scalar and fix $\tau\in(0,1)$. The parameter of interest is the $\tau$-quantile of the conditional distribution of $Y$ given $X=x_0$,
\begin{equation}\label{eq:cq-target}
    \theta_\tau(P):=\inf\{y:P\{Y\le y\}\ge \tau\}~,
\end{equation}
where $P$ denotes the conditional law of $Y$ given $X=x_0$.

Consider first the ideal experiment $S=(S_1,\ldots,S_k)\sim P^k$. Let $S_{(1)}\le\cdots\le S_{(k)}$ denote its order statistics, and set $S_{(0)}=-\infty$ and $S_{(k+1)}=+\infty$. Choose integers $0\le\ell\le u\le k+1$ such that
\begin{equation}\label{eq:cq-binomial-indices}
    b_{k,\tau}(\ell,u)
    :=\sum_{j=\ell}^{u-1}\binom{k}{j}\tau^j(1-\tau)^{k-j}
    \ge 1-\alpha~.
\end{equation}
Lemma~\ref{lem:cq-rank-coverage} in Appendix~\ref{app:conditional-quantile} gives
\[
P^k\{\theta_\tau(P)\in[S_{(\ell)},S_{(u)}]\}
\ge b_{k,\tau}(\ell,u)~.
\]
This finite-sample guarantee does not require $P$ to have a continuous CDF. Equality holds if it is continuous at $\theta_\tau(P)$; an atom at the target quantile may make the interval conservative.

The feasible interval replaces the ideal sample with the $k$ nearest outcomes. Let $S_{n,(1)}\le\cdots\le S_{n,(k)}$ denote the order statistics of the components of $S_n$, with the same endpoint conventions, and define
\begin{equation}\label{eq:cq-interval}
    C_{n,k}:=[S_{n,(\ell)},S_{n,(u)}]~.
\end{equation}

\begin{corollary}\label{cor:conditional-quantile}
For every $Q\in\mathbf Q$, $n\in\mathbf N$, and $1\le k<n$,
\begin{equation}\label{eq:cq-finite-bound}
\begin{split}
Q^n\{\theta_\tau(P)\in C_{n,k}\}
&\ge
b_{k,\tau}(\ell,u)
-\mathrm{TV}\big(\mathcal L(S_n),\mathcal L(S)\big)\\
&\ge
1-\alpha-\mathrm{TV}\big(\mathcal L(S_n),\mathcal L(S)\big)~.
\end{split}
\end{equation}
If, in addition, Assumptions~\ref{ass:g-regular}--\ref{ass:QMD} hold, Theorems~\ref{thm:high-level}--\ref{thm:H-QMD} imply
\begin{equation}\label{eq:cq-qmd-bound}
    Q^n\{\theta_\tau(P)\in C_{n,k}\}
    \ge
    1-\alpha-O\big(k^{1/2}(k/n)^{1/d}\big)~.
\end{equation}
Consequently, if $k=o(n^{2/(d+2)})$, then
\begin{equation}\label{eq:cq-asymptotic}
    \liminf_{n\to\infty}Q^n\{\theta_\tau(P)\in C_{n,k}\}
    \ge1-\alpha~.
\end{equation}
\end{corollary}

The bound in \eqref{eq:cq-finite-bound} is a finite-sample transfer statement and imposes no smoothness on the conditional law. Smoothness enters only when bounding the total variation distance. No quantile-process approximation is needed.

A closely related construction appears in \citet[Example~2.1]{fan/liu:16}. For the asymmetric $k$-nearest-neighbor estimator, their direct approach also produces an interval with order-statistic endpoints. Fan and Liu choose those endpoints using an asymptotic Gaussian shift in the quantile index and establish first-order coverage through a quantile-process limit. Under the smooth-interior conditions used to verify that limit, they impose $k=o(n^{4/(d+4)})$, but do not report a rate for the coverage error.

Our interval instead uses finite-sample binomial calibration in the ideal experiment and the explicit transfer bound in \eqref{eq:cq-finite-bound}. Under QMD, this bound accommodates boundary points and gives the sufficient condition $k=o(n^{2/(d+2)})$. Under Assumption~\ref{ass:book}, our general result also gives $k=o(n^{4/(d+4)})$, although through approximation of the joint IOS law rather than a quantile-process limit. Fan and Liu's conditions and QMD are not nested: theirs are tailored to the conditional-quantile process, whereas QMD controls the local conditional law.

\subsection{Conditional Conformal Prediction}
\label{sec:conditional-conformal}
Conformal prediction constructs a prediction set by comparing a future outcome's nonconformity score with scores from a calibration sample. Let $(X_0,Y_0)\sim Q$ denote the covariates and outcome of the unit for which we seek a prediction. Standard conformal prediction provides a set $\widehat C(X_0)$ that covers $Y_0$ with probability at least $1-\alpha$. This marginal guarantee allows for systematic undercoverage at particular values of $X_0$. A growing literature on conditional conformal prediction instead targets coverage of at least $1-\alpha$ conditional on $X_0$. See Remark~\ref{rem:conformal-examples} for empirical settings that motivate this conditional coverage objective. Recent contributions to this literature include \cite{guan:23,hore/barber:25,min/peng/zou:26}. Nontrivial, distribution-free conditional coverage cannot generally be guaranteed in finite samples without restrictions on the data-generating process \citep{barber/candes/ramdas/tibshirani:21}; our approximation results provide a conditional coverage bound under local regularity conditions.

Let the calibration sample consist of $n$ i.i.d.\ observations from $Q$. A separate training sample $\mathcal D_{\rm tr}$ may be used to construct a measurable nonconformity score conditional on $X_0 = x_0$, $v_{x_0}:\mathbf R^m\to\mathbf R$ (for example, $v_{x_0}(y)=|y-\widehat\mu(x_0)|$ when $Y$ is scalar). The two samples and the new observation $(X_0,Y_0)$ are mutually independent, so that, conditional on $X_0=x_0$, the calibration sample has law $Q^n$ and $Y_0\sim P$ is independent of both samples. We condition throughout on $\mathcal D_{\rm tr}$, so the score is fixed. For a prediction set $\widehat C(x_0)$ constructed from the calibration sample, the target is
\begin{equation}\label{eq:conformal-target}
    (Q^n\otimes P)\{Y_0\in\widehat C(x_0)\}\ge1-\alpha~,
\end{equation}
where $\alpha\in(0,1)$. This probability averages over the calibration sample and the future outcome, given $X_0=x_0$ and $\mathcal D_{\rm tr}$. The rates below are pointwise in $x_0$.

\begin{remark}[Prediction at an observed covariate value]
\label{rem:conformal-examples}
Suppose a hospital wants to predict a patient's length of stay $Y_0$ using the patient's age and initial clinical measurements $X_0$. At admission, the hospital observes $X_0=x_0$ and constructs a prediction interval for $Y_0$ at that value. An interval that achieves $1-\alpha$ coverage across all admissions may have lower coverage among patients with age and clinical measurements $x_0$. The target in \eqref{eq:conformal-target} requires coverage of at least $1-\alpha$ for these patients. Other examples include predicting loan repayment given a borrower's characteristics at application and predicting subsequent earnings given a participant's characteristics at entry into a job-training program.
\end{remark}

Consider first the ideal experiment $S=(S_1,\ldots,S_k)\sim P^k$, independent of $Y_0$. Let $T_j=v_{x_0}(S_j)$, write $T_{(1)}\le\cdots\le T_{(k)}$ for their order statistics, and set $T_{(k+1)}=+\infty$. Define
\[
j_\alpha:=\left\lceil(k+1)(1-\alpha)\right\rceil,
\qquad
\widehat C_k(x_0):=\{y:v_{x_0}(y)\le T_{(j_\alpha)}\}~.
\]
The same fixed score is applied to $k+1$ i.i.d.\ outcomes, so the scores are exchangeable. The standard conformal rank argument therefore gives
\[
(P^k\otimes P)\{Y_0\in\widehat C_k(x_0)\}\ge1-\alpha~.
\]

The feasible procedure uses $S_n$, the outcomes attached to the $k$ calibration covariates closest to $x_0$. Let $T_{n,j}=v_{x_0}(S_{n,j})$, denote their order statistics by $T_{n,(1)}\le\cdots\le T_{n,(k)}$, and set $T_{n,(k+1)}=+\infty$. The prediction set is
\begin{equation}\label{eq:split-conformal-set}
    \widehat C_{n,k}(x_0)
    :=\{y:v_{x_0}(y)\le T_{n,(j_\alpha)}\}~.
\end{equation}
The cutoff $T_{n,(j_\alpha)}$ is finite only if $k\ge(1-\alpha)/\alpha$; otherwise the set is all of $\mathbf R^m$. At $\alpha=0.05$, this requires at least $19$ neighbors. For the absolute-residual score, this is an interval centered at $\widehat\mu(x_0)$ with radius $T_{n,(j_\alpha)}$. Every selected calibration outcome is scored at the same target value $x_0$. Using $|Y_i-\widehat\mu(X_i)|$ instead requires additional regularity of the covariate-dependent score law.

\begin{corollary}\label{cor:conditional-conformal}
For the fixed measurable score $v_{x_0}$, every $Q\in\mathbf Q$, $n\in\mathbf N$, and $1\le k<n$,
\begin{equation}\label{eq:split-conformal-transfer}
    (Q^n\otimes P)\{Y_0\in\widehat C_{n,k}(x_0)\}
    \ge
    1-\alpha-\mathrm{TV}\big(\mathcal L(S_n),\mathcal L(S)\big)~.
\end{equation}
If, in addition, Assumptions~\ref{ass:g-regular}--\ref{ass:QMD} hold, Theorems~\ref{thm:high-level}--\ref{thm:H-QMD} imply
\begin{equation}\label{eq:split-conformal-rate}
    (Q^n\otimes P)\{Y_0\in\widehat C_{n,k}(x_0)\}
    \ge
    1-\alpha-O\big(k^{1/2}(k/n)^{1/d}\big)~.
\end{equation}
Consequently, if $k=o(n^{2/(d+2)})$, the limit inferior of the coverage probability is at least $1-\alpha$.
\end{corollary}

The first bound requires no smoothness beyond the setup. Theorems~\ref{thm:high-level}--\ref{thm:H-QMD} then give \eqref{eq:split-conformal-rate} directly, without approximating the score quantile. The proof is in Appendix~\ref{app:conditional-conformal}. The guarantee holds for every fixed measurable score, permits ties, and requires no accuracy condition on the fitted prediction rule; the score still affects the size and shape of the set. If Assumptions~\ref{ass:g-regular}--\ref{ass:QMD} hold at $Q_X$-almost every $x_0$, where $Q_X$ denotes the marginal law of $X$ under $Q$, the preceding bounds give coverage conditional on $X_0$ almost surely, with constants that may depend on its value.

The set in \eqref{eq:split-conformal-set} is a nearest-neighbor version of localized split conformal prediction, with equal weights and a common target-point score; see \citet{guan:23} and \citet[Section~2.1]{hore/barber:25} for the local weighting construction. Our result gives a one-sided conditional coverage guarantee at $x_0$ under local conditions on the outcome law, including at boundary points, for any fixed measurable target-point score. It follows from the joint transfer bound without a separate approximation of the score quantile.

\begin{remark}[A sharper bound for prediction]\label{rem:conformal-refinement}
Our blueprint delivers rates automatically once the ideal guarantee and approximation conditions are established. Additional structure may yield faster rates, at the cost of further analysis. For prediction, such an argument removes the factor $k^{1/2}$ from the corollary's coverage bound: instead of approximating the joint law of all $k$ calibration outcomes, it suffices to compare the law of a single future outcome.

To see this, condition on $R_{(k+1)}=r$. Up to ordering, which does not affect the prediction set, the selected calibration outcomes are then i.i.d.\ from $P_r$. If the future outcome also had law $P_r$ and were independent of these outcomes, the conformal rank argument would give coverage at least $1-\alpha$. Its actual law is $P$, so the coverage loss is at most $\mathrm{TV}(P_r,P)$. Under Assumptions~\ref{ass:g-regular}--\ref{ass:QMD}, averaging over the radius gives
\begin{equation}\label{eq:split-conformal-refined-rate}
    (Q^n\otimes P)\{Y_0\in\widehat C_{n,k}(x_0)\}
    \ge1-\alpha-O\big((k/n)^{1/d}\big)~.
\end{equation}
Thus, $k/n\to0$ suffices for the lower coverage guarantee. This single-observation comparison does not apply to our other applications: changing the local law changes the quantile target, while the two testing problems need not have a common local law under the null. The proof is in Appendix~\ref{app:conditional-conformal}; a related comparison appears in \citet[online supplementary Appendix~A.3]{hore/barber:25}.
\end{remark}

\subsection{Policy-Irrelevance Testing under Network Interference}
\label{sec:network-policy}
\citet{auerbach/guo/tabord-meehan:26} propose a nonparametric framework for causal inference with network data. Their model allows an agent's outcome to depend on the agent's treatment status, the agent's connections with other individuals in the community, and the treatment statuses of other agents in the network. The network and treatment information from the agent's perspective is called a ``rooted network,'' which we denote by $G$. Under their assumptions, a rooted network represents a policy of interest. Fix two target rooted networks, $\mathsf g_0$ and $\mathsf g_1$, and for $a\in\{0,1\}$ let $P_a$ denote the distribution of the outcome under $\mathsf g_a$. The null hypothesis of policy irrelevance is
\begin{equation}\label{eq:network-null}
H_0:P_0=P_1~.
\end{equation}

In the ideal experiment, policy irrelevance is again a two-sample problem. Throughout this subsection, $S$ and $S_n$ denote the pooled ideal and feasible samples, respectively. Let $S(0)$ and $S(1)$ be independent, with laws $P_0^k$ and $P_1^k$, respectively, and set $S:=(S(0),S(1))$. Under $H_0$, the components of $S$ are $2k$ i.i.d.\ draws from a common distribution, and hence $S$ is exchangeable. Let $\phi$ denote the permutation test of \citet{auerbach/guo/tabord-meehan:26}, whose critical values are computed by permuting the pooled sample, so that $E[\phi(S)]\le\alpha$.

The data consist of an i.i.d.\ sample of $n$ communities, each composed of a network of individuals. We implement the procedure in \citet{auerbach/guo/tabord-meehan:26} using a prespecified split of the communities into two halves of size $n/2$, assuming for simplicity that $n$ is even; the first half is used for $\mathsf g_0$ and the second for $\mathsf g_1$. In each community of the half used for $\mathsf g_a$, select the individual whose rooted network is closest to $\mathsf g_a$ according to the distance $\rho$ in \citet[Eq.~(1)]{auerbach/guo/tabord-meehan:26}, breaking ties within the community by independent uniform randomization, as in their Algorithm~1. Let $X^a$ denote the distance between the selected individual's rooted network and $\mathsf g_a$, and let $Y^a$ denote the selected individual's outcome. This produces $n/2$ ``community representatives.'' Let $S_n(a)$ collect the outcomes of the $k$ representatives with the smallest values of $X^a$, listed in their original sample order. The resulting feasible experiment is $S_n:=(S_n(0),S_n(1))$, and applying the same permutation rule to $S_n$ gives the feasible test $\phi(S_n)$. Here $Q$ denotes the law of a community together with its tie-breaking draws.

To connect this construction with our framework, let $Q_a$ be the law of the representative pair $(X^a,Y^a)$ for $a\in\{0,1\}$. Within each half, the representative pairs are i.i.d.\ draws from $Q_a$, and the two halves are independent. Hence $S_n(a)$ is the IOS vector attached to the $k$ smallest values of $X^a$ in a sample of size $n/2$ from $Q_a$. For $r>0$ with positive marginal density under $Q_a$, let $P_{a,r}$ denote the conditional law of $Y^a$ given $X^a=r$, defined through the densities as in Section~\ref{sec:two-experiments}. Since distances are nonnegative, each half is a one-dimensional IOS experiment with conditioning point zero, a boundary point of its covariate support.

\begin{corollary}\label{cor:network-policy}
Suppose that, for each $a\in\{0,1\}$, $Q_a$ satisfies Assumptions~\ref{ass:g-regular}--\ref{ass:QMD} at zero with $P_a$ as the conditional law of $Y^a$ at $X^a=0$. Then, uniformly over even $n\in\mathbf N$ and $1\le k<n/2$, Theorems~\ref{thm:high-level}--\ref{thm:H-QMD} imply
\begin{equation}\label{eq:EricEtAl_S1}
\mathrm{TV}\big(\mathcal L(S_n),\mathcal L(S)\big)=O\big(k^{3/2}/n\big)~.
\end{equation}
Under $H_0$, it follows that
\begin{equation}\label{eq:EricEtAl_S2}
E_{Q^n}[\phi(S_n)]\le\alpha+O\big(k^{3/2}/n\big)~.
\end{equation}
Consequently, asymptotic size control holds if $k=o(n^{2/3})$.
\end{corollary}

The proof, given in Appendix~\ref{app:application-proofs}, applies the IOS approximation separately to the two halves and combines the bounds using the product inequality for total variation. As in Corollary~\ref{cor:rdd-balance}, requiring $P_a$ to be the conditional law at zero identifies the limiting law of the selected outcome with the target law, and the common $2/3$ exponent reflects that both applications reduce to a one-dimensional boundary experiment under QMD.

\citet[Theorem~4]{auerbach/guo/tabord-meehan:26} establish asymptotic validity for fixed $k$ under their Assumptions~1, 5, and~7. Their simulation evidence suggests the diverging choice $k=2\lfloor\log n\rfloor$ in our notation (p.~192). This choice satisfies our growth condition and yields a size bound of $\alpha+O((\log n)^{3/2}/n)$ under our assumptions. Our result does not cover the data-dependent alternating partition they recommend in practice because that partition need not produce the independent subsamples used in our proof.

\section{Concluding Remarks}\label{sec:conclusions}

The paper provides a three-step blueprint for inference with nearest outcomes. First, establish a finite-sample guarantee in an ideal experiment consisting of independent draws from the conditional law at the target point. Second, apply the same procedure to the outcomes attached to the nearest observed covariates. Third, bound the total variation distance between the two experiments. The ideal guarantee then carries over to the feasible procedure, with an error no larger than this distance.

Our approximation results give sufficient restrictions on how fast the number of selected neighbors may grow with the sample size. Under quadratic mean differentiability and local regularity of the covariate distribution, the local Hellinger and total variation errors are at most linear in the neighborhood radius, including at boundary points. Theorem~\ref{thm:H-QMD} also establishes the limits of these rates. A boundary example attains the linear rate in both distances. At interior points, both the Hellinger and total variation errors are smaller than linear for every law satisfying the assumptions. Nevertheless, a single such interior model rules out every positive polynomial improvement on the linear rate in either distance. Interior location alone therefore does not deliver a faster polynomial rate under these conditions.

The four applications combine the same transfer argument with different finite-sample guarantees: permutation validity, binomial calibration, and conformal ranks. Controlling the joint law of the nearest outcomes allows the approximation to be used across procedures, without a separate analysis of each statistic. Additional structure can yield sharper bounds, as the prediction application illustrates, at the cost of further analysis.

\appendices

\renewcommand{\theequation}{\Alph{section}-\arabic{equation}}
\setcounter{equation}{0}

\section{Proof of the main results}\label{app:main-proofs}

For probability measures $M$ and $N$ with densities $p_M$ and $p_N$ relative to a common dominating measure $\lambda$, we use the normalization
\[
\mathrm H^2(M,N):=\frac12\int(\sqrt{p_M}-\sqrt{p_N})^2\,d\lambda,
\qquad
\mathrm{TV}(M,N):=\frac12\int|p_M-p_N|\,d\lambda~.
\]

\subsection{Proof of Theorem \ref{thm:high-level}}

\underline{Part (a).} Construct $S\sim P^k$ independently of the observed sample. Let $r_1$ be as in Lemma~\ref{lem:Fr}, reduced if necessary so that the assumed local Hellinger and total variation bounds hold for $0<r\le r_1$. All constants below may depend on the fixed law and the exponents, but not on $n$ or $k$. Lemma \ref{lem:Fr} implies that $F(r) \geq F( r/2) >0$ for all $r\in ( 0,r_{1}]$. Then, \citet[Theorem 1]{kaufmann/reiss:92} implies that, for $\mathcal L(R_{(k+1)})$-almost every $r\in(0,r_1)$,
\begin{equation}\label{eq:HL_1}
    \mathcal{L}( S_{n}|R_{( k+1) }=r) =P_{r}^{k}~,
\end{equation}
where $P_{r}^{k}$ denotes the $k$-fold product measure of $P_{r}$.

Consider the following derivation for any $n\in\mathbf N$ and $1\le k<n$,
\begin{align}
\mathrm{H}^{2}&(\mathcal{L}( S_{n}) ,\mathcal{L}( S) ) \overset{(1)}{\leq }\int_{0}^{\infty }\mathrm{H}^{2}( \mathcal{L} ( S_{n}|R_{(k+1) }=r) ,\mathcal{L}( S|R_{(k+1) }=r) ) dP_{R_{(k+1) }}( r)  \notag\\
&\overset{(2)}{=}\int_{0}^{r_{1}}\mathrm{H}^{2}( P_{r}^{k},P^{k}) dP_{R_{(k+1) }}( r) +\int_{r_{1}}^{\infty }\mathrm{H}^{2}( \mathcal{L}( S_{n}|R_{(k+1) }=r) ,\mathcal{L}( S|R_{(k+1) }=r)) dP_{R_{(k+1) }}( r)\notag \\
&\overset{(3)}{\leq }k\int_{0}^{r_{1}}\mathrm{H}^{2}( P_{r},P) dP_{R_{(k+1) }}( r) +P( R_{(k+1) }\geq r_{1}) ~,\label{eq:high-level_1}
\end{align}
where (1) holds by joint convexity of the squared Hellinger distance (e.g., \citet[Lemma 3.1.3]{reiss:1993}), applied to $\mathcal L(S_n)=\int\mathcal L(S_n|R_{(k+1)}=r)\,dP_{R_{(k+1)}}(r)$ and $\mathcal L(S)=\int\mathcal L(S|R_{(k+1)}=r)\,dP_{R_{(k+1)}}(r)$, (2) by \eqref{eq:HL_1}, $S\perp R_{(k+1) }$, and $S\sim P^{k}$, and (3) by $\mathrm{H}^{2}(P',P'') \leq 1$ and Bernoulli's inequality, which states that $1-( 1-x) ^{k}\leq kx$ for any $x\in [ 0,1] $ and $k\in \mathbf{N}$.

Given \eqref{eq:high-level_1}, we complete the proof by finding constants $C_1,C_2 \in (0,\infty)$ such that, for all $n\in\mathbf N$ and $1\le k<n$,
\begin{align}
    \int_{0}^{r_{1}}\mathrm{H}^{2}( P_{r},P) dP_{R_{(k+1) }}( r) &\leq C_{1}(k/n)^{2a_{h}/d}\label{eq:high-level_2}\\
    P( R_{(k+1) }\geq r_{1})& \leq C_{2}k(k/n)^{2a_{h}/d}~.\label{eq:high-level_3}
\end{align}

We begin by proving \eqref{eq:high-level_2}. By $\mathrm{H}(P_r,P)=O(r^{a_h})$, there is $C_3<\infty$ such that
\begin{equation}\label{eq:high-level_4}
    \mathrm{H}^2(P_r,P)\le C_3r^{2a_h}
    \qquad\text{for all }0<r\le r_1~.
\end{equation}
Set $\beta:=2a_h/d$. If $\beta=0$, \eqref{eq:high-level_2} and \eqref{eq:high-level_3} follow immediately from $\mathrm H^2(P_r,P)\le1$ and $P\{R_{(k+1)}\ge r_1\}\le1\le k$. Suppose, therefore, that $\beta>0$. For every $u\in(0,1)$,
\begin{equation}\label{eq:truncated-quantile-bound}
    \big(F^{-1}(u)\big)^{2a_h}
    I\{F^{-1}(u)\le r_1\}
    \le C_L^{-\beta}u^\beta~.
\end{equation}
To see this, if $u\le C_Lr_1^d$, Lemma~\ref{lem:Fr} gives $F^{-1}(u)\le(u/C_L)^{1/d}$ and so $(F^{-1}(u))^{2a_h}\le C^{-\beta}_L u^{\beta}$. If $u>C_Lr_1^d$ and the indicator in \eqref{eq:truncated-quantile-bound} equals one, then $F^{-1}(u)\le r_1<(u/C_L)^{1/d}$; if the indicator equals zero, the inequality is immediate. It follows that
\begin{align*}
    \int_0^{r_1}\mathrm H^2(P_r,P)\,dP_{R_{(k+1)}}(r)
    &\le C_3E\left[R_{(k+1)}^{2a_h}I\{R_{(k+1)}\le r_1\}\right]\\
    &=C_3E\left[\big(F^{-1}(U_{(k+1:n)})\big)^{2a_h}
      I\{F^{-1}(U_{(k+1:n)})\le r_1\}\right]\\
    &\le C_3C_L^{-\beta}E\left[U_{(k+1:n)}^\beta\right]\\
    &\le C_4(k/n)^\beta~.
\end{align*}
The last inequality follows from Lemma~\ref{lem:ordered_uniform_power}, applied to the $(k+1)$st order statistic, and $k+1\le2k$. This proves \eqref{eq:high-level_2}.

We now prove \eqref{eq:high-level_3}. Let $r_\star:=r_1/2$ and $p_\star:=F(r_\star)>0$, where positivity follows from Lemma~\ref{lem:Fr}. Define
\[
    N_\star:=\sum_{i=1}^n I\{R_i<r_\star\}~.
\]
Then $N_\star\sim\mathrm{Binomial}(n,p_\star)$ and
\[
    \{R_{(k+1)}\ge r_1\}
    \subseteq \{R_{(k+1)}\ge r_\star\}
    =\{N_\star\le k\}~.
\]
If $k/n\le p_\star/2$, the Chernoff bound for the lower tail of a binomial random variable gives
\begin{align}
    P\{R_{(k+1)}\ge r_1\}\le 
    P\{N_\star\le k\}
    \le P\{N_\star\le np_\star/2\}
    \le \exp(-np_\star/8)
    \le C_5n^{-\beta}
    \le C_5k(k/n)^\beta~,
    \label{eq:high-level_5}
\end{align}
for some $C_5<\infty$ which does not depend on $n$ or $k$, as exponential decay dominates polynomial decay. The last inequality uses $k\ge1$. If $k/n>p_\star/2$, then
\begin{equation}
    P\{R_{(k+1)}\ge r_1\}
    \le1
    \le \left(\frac{2}{p_\star}\right)^\beta(k/n)^\beta
    \le \left(\frac{2}{p_\star}\right)^\beta k(k/n)^\beta~.
    \label{eq:high-level_6}
\end{equation}
Combining the two cases proves \eqref{eq:high-level_3} with $C_2:=\max\{C_5,(2/p_\star)^\beta\}$. Combining \eqref{eq:high-level_1}--\eqref{eq:high-level_3}, including the case $\beta=0$, gives
\[
\mathrm H^2(\mathcal L(S_n),\mathcal L (S))
\le (C_1+C_2)k(k/n)^{2a_h/d}~.
\]
Taking square roots proves part (a).

\noindent\underline{Part (b).} By part (a) and \eqref{eq:HvsTV}, we obtain that, for any $n\in\mathbf N$ and $1\le k<n$, $\mathrm{TV}(\mathcal{L}(S_{n}),\mathcal{L}(S))=O(k^{1/2} (k/n)^{a_{h}/d})$. To complete the proof, it then suffices to show that, for any $n\in\mathbf N$ and $1\le k<n$, we also have $\mathrm{TV}(\mathcal{L}(S_{n}),\mathcal{L}(S))=O(k(k/n)^{a_{tv}/d})$. 

The argument used to derive \eqref{eq:HL_1} still applies. Thus, for any $n\in\mathbf N$ and $1\le k<n$, we have
\begin{align}
    \mathrm{TV}& (\mathcal{L}(S_{n}),\mathcal{L}(S))  \overset{(1)}{=}\sup_{B\in \mathcal{B}}\left\vert \int_{0}^{\infty }\left( P\left( S_{n}\in B|R_{(k+1)}=r\right) -P\left( S\in B|R_{(k+1)}=r\right) \right) dP_{R_{(k+1)}}(r)\right\vert \notag \\
    & \leq \int_{0}^{\infty }\sup_{B\in \mathcal{B}}\left\vert P\left( S_{n}\in B|R_{(k+1)}=r\right) -P\left( S\in B|R_{(k+1)}=r\right) \right\vert dP_{R_{(k+1)}}(r) \notag \\
    & =\int_{0}^{\infty }\mathrm{TV}(\mathcal{L}(S_{n}|R_{(k+1)}=r),\mathcal{L}(S|R_{(k+1)}=r))dP_{R_{(k+1)}}(r)\notag \\
    & \overset{(2)}{=}\int_{0}^{r_{1}}\mathrm{TV}(P_{r}^{k},P^{k})dP_{R_{(k+1)}}(r)+\int_{r_{1}}^{\infty }\mathrm{TV}(\mathcal{L}(S_{n}|R_{(k+1)}=r),\mathcal{L}(S|R_{(k+1)}=r))dP_{R_{(k+1)}}(r) \notag \\
    & \overset{(3)}{\leq }k\int_{0}^{r_{1}}\mathrm{TV}(P_{r},P)dP_{R_{(k+1)}}(r)+P(R_{(k+1)}\geq r_{1})~,  \label{eq:high-level_1b}
\end{align}
where (1) holds by the law of total probability and using $\mathcal{B}$ to denote the common $\sigma$-algebra to $\mathcal{L}(S_{n})$ and $\mathcal{L}(S)$, and (2) by \eqref{eq:HL_1}, $ S\perp R_{(k+1)}$, and $S\sim P^{k}$, and (3) by \citet[Eq.\ (4.5)]{hoeffding/wolfowitz:1958} and $\mathrm{TV}\le1$. Given \eqref{eq:high-level_1b}, we complete the proof by finding constants ${C_{1}^{\prime},C_{2}^{\prime}}\in (0,\infty )$ such that, for all $n\in\mathbf N$ and $1\le k<n$, 
\begin{align}
\int_{0}^{r_{1}}\mathrm{TV}(P_{r},P)dP_{R_{(k+1)}}(r)& \leq {C_{1}^{\prime}}(k/n)^{a_{tv}/d} \label{eq:high-level_2b}\\
P(R_{(k+1)}\geq r_{1})& \leq {C_{2}^{\prime}}k(k/n)^{a_{tv}/d}~.\label{eq:high-level_3b}
\end{align}
The bounds in \eqref{eq:high-level_2b} and \eqref{eq:high-level_3b} follow by repeating the preceding argument with $\mathrm H(P_r,P)=O(r^{a_h})$ replaced by $\mathrm{TV}(P_r,P)=O(r^{a_{tv}})$ and $\beta:=a_{tv}/d$. Combining this total variation bound with the bound from part (a), and taking the larger of their constants, proves the stated minimum.

\subsection{Proof of Theorem \ref{thm:H-QMD}}

Let $r_1$ be as in Lemma~\ref{lem:Fr}. For $x\in\mathcal X$ sufficiently close to $x_0$, the density $g(x)$ is positive, so $p_x$ and $D(x)$ are defined. Lemma~\ref{lem:Dx} with Cauchy-Schwarz inequality give, as $x\to x_0$ through these points,
\begin{align}\label{eq:thm3_1}
\begin{aligned}
    D(x)
&= \tfrac14 \int ((x-x_0)^{\top} \dot{\ell}_{x_0}(y))^2 p_{x_0}(y) \nu (dy)+o(\|x-x_0\|^2)\\
&\leq \tfrac14  \|x-x_0\|^2\int \| \dot{\ell}_{x_0}(y)\|^2 p_{x_0}(y) \nu (dy)+o(\|x-x_0\|^2)~.
\end{aligned}
\end{align}
Consequently, there are $C<\infty$ and $\bar r>0$ such that, for every $x\in\mathcal X$ with $\|x-x_0\|<\bar r$,
\begin{equation}\label{eq:thm3_3}
D(x)\le C\|x-x_0\|^2~.
\end{equation}
By Lemma~\ref{lem:Fr}, for $0<r\le r_1$,
\begin{equation}\label{eq:thm3_4}
F(r)=\int_{B_r\cap\mathcal X}g(x)\,dx\ge C_Lr^d>0~.
\end{equation}
Thus the probability measure
\begin{equation}\label{eq:thm3_5}
\mu_r(dx):=\frac{g(x)I\{x\in B_r\cap\mathcal X\}}{F(r)}\,dx~,
\end{equation}
satisfies $P_r=\int P_x\,\mu_r(dx)$. Let $r_2 := \min\{r_1,\bar{r}\}$. For any $r\in ( 0,r_{2})$, consider the following derivation.
\begin{align}
\mathrm{H}^{2}(P_{r},P)
&\overset{(1)}{\leq }\int \mathrm{H}^{2}(P_{x},P)\mu _{r}(dx)\notag\\
&\overset{(2)}{=}\frac{1}{2F(r)}\int_{B_{r}\cap \mathcal{X}}D(x)g(x)dx  \notag \\
&\overset{(3)}{\leq }\frac{C}{2F(r)}\int_{B_{r}\cap \mathcal{X}}\Vert x-x_{0}\Vert ^{2}g(x)dx  \notag \\
&\overset{(4)}{\leq }\frac{C}{2C_{L}r^{d}}( g( x_{0}) +C_{g}r) \int_{B_{r}\cap \mathcal{X}}\Vert x-x_{0}\Vert ^{2}dx  \notag \\
&\overset{(5)}{=}O(r^{2})~,  \label{eq:thm3_6} 
\end{align}
where (1) holds by the convexity lemma (e.g., \citet[Lemma 3.1.3]{reiss:1993}), with measure $\mu _{r}(dx)$ and $P_{r}=\int P_{x}\mu _{r}(dx)$, (2) by $\mathrm{H}^{2}(P_{x},P)=D(x)/2$ and \eqref{eq:thm3_5}, (3) by \eqref{eq:thm3_3}, (4) by \eqref{eq:thm3_4} and Assumption~\ref{ass:g-regular}, and (5) by $\int_{B_{r}\cap \mathcal X}\Vert x-x_{0}\Vert ^{2}dx=O(r^{d+2})$. From \eqref{eq:thm3_6}, we obtain
\begin{equation*}
\mathrm{H}(P_{r},P)=O(r)~.
\end{equation*}
The convergence rate of the total variation distance follows from its relationship with the Hellinger distance, $\mathrm{TV}(P_{r},P)\leq \sqrt{2} \mathrm{H}(P_{r},P)$. This completes the proof of \eqref{eq:H-TV-QMD}. 

Part (i) of the Theorem follows from Lemma~\ref{lem:QMD-sharpness}. The first statement in Part (ii) follows from Lemma~\ref{lem:QMD-TV-or}, and the second follows from Lemma~\ref{lem:QMD-no-poly-improvement}. This completes the proof. \qed

\begin{proof}[Proof of Remark~\ref{rem:interior-joint}]
Fix $Q$ and a sequence $k=k_n$ as in the remark. Let $r_1$ be as in Lemma~\ref{lem:Fr}. Given $\epsilon>0$, Theorem~\ref{thm:H-QMD}(ii) gives $\delta\in(0,r_1)$ such that $\mathrm H^2(P_r,P)\le\epsilon r^2$ for $0<r<\delta$. Applying the argument in \eqref{eq:high-level_1} with $\delta$ in place of $r_1$ yields
\begin{align*}
    \mathrm H^2\big(\mathcal L(S_n),\mathcal L(S)\big)
    &\le \epsilon k E_{Q^n}\!\left[R_{(k+1)}^2 I\{R_{(k+1)}<\delta\}\right]
       +Q^n\{R_{(k+1)}\ge\delta\}\\
    &\le \epsilon C k(k/n)^{2/d}
       +Q^n\{R_{(k+1)}\ge\delta\}~.
\end{align*}
The second inequality follows from \eqref{eq:truncated-quantile-bound} with $a_h=1$ and Lemma~\ref{lem:ordered_uniform_power}, since $\delta<r_1$. The constant $C$ does not depend on $\epsilon$, $\delta$, $n$, or $k$.

Let $p_\delta:=F(\delta/2)>0$. Since $k/n\to0$, eventually $k/n\le p_\delta/2$. The binomial argument in \eqref{eq:high-level_5} then gives
\[
    Q^n\{R_{(k+1)}\ge\delta\}\le \exp(-np_\delta/8)~.
\]
Since $k\ge1$, we have $k(k/n)^{2/d}\ge n^{-2/d}$, and hence
\[
    \frac{Q^n\{R_{(k+1)}\ge\delta\}}{k(k/n)^{2/d}}
    \le n^{2/d}\exp(-np_\delta/8)\to 0~.
\]
It follows that
\[
    \limsup_{n\to\infty}
    \frac{\mathrm H^2\big(\mathcal L(S_n),\mathcal L(S)\big)}
    {k(k/n)^{2/d}}\le \epsilon C~.
\]
Letting $\epsilon\downarrow0$ proves the Hellinger claim. The total variation claim follows from \eqref{eq:HvsTV}, and $k=O(n^{2/(d+2)})$ makes $k^{1/2}(k/n)^{1/d}$ bounded.
\end{proof}

\subsection{Proofs for the Applications}\label{app:application-proofs}

\subsubsection{Covariate balance testing in regression discontinuity designs}

\begin{proof}[Proof of Corollary~\ref{cor:rdd-balance}]
Let
\[
N_{n,-}:=\sum_{i=1}^n1\{X_i<0\},
\qquad
N_{n,+}:=\sum_{i=1}^n1\{X_i>0\},
\]
and write $\pi_-:=Q\{X<0\}$ and $\pi_+:=Q\{X>0\}$. Let $\mathcal G_n$ be the $\sigma$-field generated by the signs of $X_1,\ldots,X_n$. Conditional on $\mathcal G_n$, the observations on the two sides are independent samples from $Q_-$ and $Q_+$. On the event
\[
\mathcal E_n:=\{N_{n,-}\ge \pi_-n/2,\ N_{n,+}\ge \pi_+n/2\},
\]
we have $q<N_{n,-}\wedge N_{n,+}$ for all sufficiently large $n$. Use the selected vector whenever both sides contain at least $q$ observations, and define $S_n$ to be a fixed vector otherwise. Since $q=o(n)$, the two conditional applications of Theorem~\ref{thm:high-level} below are valid on $\mathcal E_n$ for all sufficiently large $n$. Theorems~\ref{thm:high-level}--\ref{thm:H-QMD}, applied conditionally and separately to the two sides, and the product inequality for total variation give
\begin{align*}
\mathrm{TV}\!\left(\mathcal L(S_n\mid\mathcal G_n),
P_-^q\otimes P_+^q\right)&\le
C_-\frac{q^{3/2}}{N_{n,-}}+C_+\frac{q^{3/2}}{N_{n,+}}\le
\left(\frac{2C_-}{\pi_-}+\frac{2C_+}{\pi_+}\right)\frac{q^{3/2}}{n}
=O(q^{3/2}/n)~,
\end{align*}
uniformly on $\mathcal E_n$. Binomial concentration implies that $Q^n\{\mathcal E_n^c\}\le 2\exp(-cn)$ for some $c>0$. Convexity of total variation over the conditional laws therefore yields
\[
\mathrm{TV}\!\left(\mathcal L(S_n),P_-^q\otimes P_+^q\right)
\le O(q^{3/2}/n)+2\exp(-cn)
=O(q^{3/2}/n)~.
\]
Under $H_0$, the $2q$ pooled ideal outcomes are i.i.d.\ from a common law and hence exchangeable. The ideal permutation rule therefore has rejection probability at most $\alpha$. The variational characterization of total variation gives
\[
E_{Q^n}[\phi(S_n)]
\le \alpha+O(q^{3/2}/n)~,
\]
which proves \eqref{eq:rdd-size-bound}. Finally, $q=o(n^{2/3})$ makes the approximation error vanish.
\end{proof}

\subsubsection{Conditional quantile inference}\label{app:conditional-quantile}

The next lemma gives the ideal-experiment coverage statement without imposing continuity of $P$. Let $S_1,\ldots,S_k$ be i.i.d.\ with common law $P$ and CDF $F$, write $S_{(1)}\le\cdots\le S_{(k)}$ for their order statistics, and set $S_{(0)}=-\infty$ and $S_{(k+1)}=+\infty$.

\begin{lemma}\label{lem:cq-rank-coverage}
Fix $\tau\in(0,1)$ and let $\theta_\tau:=\inf\{y:F(y)\ge \tau\}~.$ For integers $0\le\ell\le u\le k+1$,
\begin{equation}\label{eq:cq-ideal-coverage}
P^k\{\theta_\tau\in[S_{(\ell)},S_{(u)}]\}
\ge b_{k,\tau}(\ell,u)
:=P\{\ell\le B_{k,\tau}\le u-1\},
\end{equation}
where $B_{k,\tau}\sim\mathrm{Binomial}(k,\tau)$. If $F$ is continuous at $\theta_\tau$, equality holds in \eqref{eq:cq-ideal-coverage}.
\end{lemma}

\begin{proof}
Define
\[
\tau_-:=F(\theta_\tau-)=P\{S_i<\theta_\tau\},
\qquad
\tau_+:=F(\theta_\tau)=P\{S_i\le\theta_\tau\}~.
\]
The definition of $\theta_\tau$ implies $\tau_-\le\tau\le\tau_+$. Let
\[
N_{<}:=\sum_{i=1}^k1\{S_i<\theta_\tau\},
\qquad
N_{\le}:=\sum_{i=1}^k1\{S_i\le\theta_\tau\}~.
\]
The interval covers $\theta_\tau$ if and only if
\begin{equation}\label{eq:cq-cover-event}
\{\theta_\tau\in[S_{(\ell)},S_{(u)}]\}
=\{N_{\le}\ge\ell,\ N_{<}\le u-1\}~.
\end{equation}
Since $N_{\le}\sim\mathrm{Binomial}(k,\tau_+)$ and $N_{<}\sim\mathrm{Binomial}(k,\tau_-)$, stochastic ordering and the union bound give
\begin{align}
P^k\{\theta_\tau\notin[S_{(\ell)},S_{(u)}]\}
&\le P^k\{N_{\le}\le\ell-1\}+P^k\{N_{<}\ge u\}\notag\\
&\le \sum_{j=0}^{\ell-1}\binom{k}{j}\tau^j(1-\tau)^{k-j}
+\sum_{j=u}^{k}\binom{k}{j}\tau^j(1-\tau)^{k-j}\notag\\
&=1-b_{k,\tau}(\ell,u)~.\label{eq:cq-tail-order}
\end{align}
Taking complements proves \eqref{eq:cq-ideal-coverage}. If $F$ is continuous at $\theta_\tau$, then $\tau_-=\tau_+=\tau$. In this case $N_{<}=N_{\le}$ almost surely, this common count has the $\mathrm{Binomial}(k,\tau)$ distribution, and \eqref{eq:cq-cover-event} gives equality.
\end{proof}

\begin{proof}[Proof of Corollary~\ref{cor:conditional-quantile}]
Let
\[
\mathcal A_\tau:=\{s\in\mathbf R^k:s_{(\ell)}\le\theta_\tau(P)\le s_{(u)}\}~,
\]
where the conventions for the zeroth and $(k+1)$st order statistics are as above. The variational characterization of total variation implies
\begin{equation}\label{eq:cq-tv-transfer}
Q^n\{S_n\in\mathcal A_\tau\}
\ge P^k\{S\in\mathcal A_\tau\}
-\mathrm{TV}\big(\mathcal L(S_n),\mathcal L(S)\big)~.
\end{equation}
Lemma~\ref{lem:cq-rank-coverage} and \eqref{eq:cq-binomial-indices} bound the first term on the right-hand side from below by $b_{k,\tau}(\ell,u)\ge1-\alpha$, which proves \eqref{eq:cq-finite-bound}. Under Assumptions~\ref{ass:g-regular}--\ref{ass:QMD}, Theorems~\ref{thm:high-level}--\ref{thm:H-QMD} give
\[
\mathrm{TV}\big(\mathcal L(S_n),\mathcal L(S)\big)
=O\big(k^{1/2}(k/n)^{1/d}\big)~.
\]
Substituting this bound into \eqref{eq:cq-finite-bound} proves \eqref{eq:cq-qmd-bound}. Finally, $k^{1/2}(k/n)^{1/d}=o(1)$ if and only if $k=o(n^{2/(d+2)})$, which gives \eqref{eq:cq-asymptotic}.
\end{proof}

\subsubsection{Conditional conformal prediction}\label{app:conditional-conformal}

Throughout this subsection, we condition on $X_0=x_0$ and the training sample $\mathcal D_{\rm tr}$, so that the calibration sample has law $Q^n$, the independent future outcome has law $P$, and the score $v_{x_0}$ is fixed.

\begin{proof}[Proof of Corollary~\ref{cor:conditional-conformal}]
Fix the score $v_{x_0}$. For $s=(s_1,\ldots,s_k)$, let $t_{(j)}(s)$ denote the $j$th order statistic of $v_{x_0}(s_1),\ldots,v_{x_0}(s_k),$ with $t_{(k+1)}(s)=+\infty$, and define
\[
A
:=
\left\{
(s,y)\in(\mathbf R^m)^k\times\mathbf R^m:
v_{x_0}(y)\le t_{(j_\alpha)}(s)
\right\}~.
\]
For any probability law $G$ on $\mathbf R^m$, consider $k+1$ i.i.d.\ draws from $G$ and apply the same score $v_{x_0}$ to each. Break score ties using independent $\mathrm{Unif}(0,1)$ variables. The randomized rank of the $(k+1)$st score is then uniform on $\{1,\ldots,k+1\}$, and rank at most $j_\alpha$ implies membership in $A$. Hence,
\begin{equation}\label{eq:conformal-iid-rank}
(G^k\otimes G)(A)
\ge
\frac{j_\alpha}{k+1}
\ge
1-\alpha~.
\end{equation}
This argument allows ties. If $j_\alpha=k+1$, then $A$ is the entire sample space, which gives the small-$k$ qualification following Corollary~\ref{cor:conditional-conformal}.

Since $Y_0\sim P$ is independent of both $S_n$ and the ideal sample $S\sim P^k$, the feasible and ideal coverage probabilities are $(\mathcal L(S_n)\otimes P)(A)$ and $(P^k\otimes P)(A),$ respectively. By the variational characterization of total variation and the fact that tensoring with a common probability measure preserves total variation,
\begin{align*}
(\mathcal L(S_n)\otimes P)(A)
&\ge
(P^k\otimes P)(A)
-
\mathrm{TV}\big(\mathcal L(S_n)\otimes P,P^k\otimes P\big)\\
&=
(P^k\otimes P)(A)
-
\mathrm{TV}\big(\mathcal L(S_n),P^k\big)\\
&\ge
1-\alpha-
\mathrm{TV}\big(\mathcal L(S_n),P^k\big)~,
\end{align*}
where the last inequality follows from \eqref{eq:conformal-iid-rank} with $G=P$. This proves \eqref{eq:split-conformal-transfer}. Under Assumptions~\ref{ass:g-regular}--\ref{ass:QMD}, Theorems~\ref{thm:high-level}--\ref{thm:H-QMD} give
\[
\mathrm{TV}\big(\mathcal L(S_n),P^k\big)
=
O\big(k^{1/2}(k/n)^{1/d}\big)~.
\]
Substituting this bound into \eqref{eq:split-conformal-transfer} proves \eqref{eq:split-conformal-rate}. Finally, $k=o\big(n^{2/(d+2)}\big)
$ implies $k^{1/2}(k/n)^{1/d}\to0$.
\end{proof}

\begin{proof}[Proof of Remark~\ref{rem:conformal-refinement}]
Use the set $A$ from the preceding proof, and let $r_1$ be as in Lemma~\ref{lem:Fr}, reduced if necessary so that $\mathrm{TV}(P,P_r)\le Cr$ for all $0<r\le r_1$, as guaranteed by Theorem~\ref{thm:H-QMD}. By Lemma~\ref{lem:Fr} and monotonicity of $F$, $F(r)>0$ for every $r>0$, so \citet[Theorem~1]{kaufmann/reiss:92} gives the representation in \eqref{eq:HL_1} for $\mathcal L(R_{(k+1)})$-almost every $r>0$, not only for $r<r_1$. By \eqref{eq:HL_1}, conditional on $R_{(k+1)}=r$, the calibration sample has law $P_r^k$, while the independent future outcome $Y_0$ has law $P$. Hence, for every such $r$, the conditional coverage probability is $(P_r^k\otimes P)(A).$ Under the auxiliary experiment $P_r^{k+1}$, the $k+1$ observations are i.i.d.\ from $P_r$. Therefore, \eqref{eq:conformal-iid-rank}, applied with $G=P_r$, gives
\[
(P_r^{k+1})(A)\ge1-\alpha~.
\]
The actual and auxiliary experiments differ only in the law of the future outcome. Thus,
\begin{align*}
(P_r^k\otimes P)(A)
&\ge
(P_r^k\otimes P_r)(A)
-
\mathrm{TV}(P_r^k\otimes P,P_r^k\otimes P_r)\\
&\ge
1-\alpha-\mathrm{TV}(P,P_r)~,
\end{align*}
where the equality $\mathrm{TV}(P_r^k\otimes P,P_r^k\otimes P_r)
=
\mathrm{TV}(P,P_r)$ again follows from invariance of total variation under tensoring with a common probability measure. Integrating over $R_{(k+1)}$ gives
\begin{equation}\label{eq:conformal-local-transfer}
(Q^n\otimes P)\{Y_0\in\widehat C_{n,k}(x_0)\}
\ge
1-\alpha-
E_{Q^n}\!\left[
\mathrm{TV}(P,P_{R_{(k+1)}})
\right]~.
\end{equation}
It remains to bound the final expectation. The arguments used in the proof of Theorem~\ref{thm:high-level}, with $a_{tv}=1$, give
\begin{align*}
E_{Q^n}\!\left[
\mathrm{TV}(P,P_{R_{(k+1)}})
\right]
&\le
\int_0^{r_1}
\mathrm{TV}(P,P_r)\,dP_{R_{(k+1)}}(r)
+
P\{R_{(k+1)}\ge r_1\}=O((k/n)^{1/d})~.
\end{align*}
The first term is $O((k/n)^{1/d})$ by \eqref{eq:high-level_2b} with $a_{tv}=1$. For the second term, let $p_\star:=F(r_1/2)>0$ and argue as in \eqref{eq:high-level_5}--\eqref{eq:high-level_6} with $\beta=1/d$, but without the factor $k$. If $k/n\le p_\star/2$, then $P\{R_{(k+1)}\ge r_1\}\le\exp(-np_\star/8)\le Cn^{-1/d}\le C(k/n)^{1/d}$ for a constant $C<\infty$ that does not depend on $n$ or $k$. If $k/n>p_\star/2$, then $P\{R_{(k+1)}\ge r_1\}\le1\le(2/p_\star)^{1/d}(k/n)^{1/d}$. Hence $P\{R_{(k+1)}\ge r_1\}=O((k/n)^{1/d})$ uniformly over $n\in\mathbf N$ and $1\le k<n$. Substituting this bound into \eqref{eq:conformal-local-transfer} proves \eqref{eq:split-conformal-refined-rate}. Finally, $k/n\to0$ makes the coverage error vanish.
\end{proof}

\subsubsection{Policy-irrelevance testing under network interference}

\begin{proof}[Proof of Corollary~\ref{cor:network-policy}]
Fix $a\in\{0,1\}$. Within the half used for $\mathsf g_a$, the representative pairs $(X^a,Y^a)$ are $n/2$ i.i.d.\ draws from $Q_a$, and $S_n(a)$ is the corresponding IOS vector with conditioning point zero. Since Assumptions~\ref{ass:g-regular} and \ref{ass:QMD} hold at zero for $Q_a$ with $P_a$ as the conditional law of $Y^a$ at $X^a=0$, Theorems~\ref{thm:high-level}--\ref{thm:H-QMD}, applied with $d=1$ and $a_h=a_{tv}=1$, give
\begin{equation}\label{eq:EricEtAl_1}
\mathrm{TV}\big(\mathcal L(S_n(a)),P_a^k\big)=O\big(k^{3/2}/(n/2)\big)=O\big(k^{3/2}/n\big)
\end{equation}
uniformly over even $n\in\mathbf N$ and $1\le k<n/2$. The two halves consist of disjoint sets of independent communities, together with their tie-breaking draws, so $S_n(0)$ and $S_n(1)$ are independent. The product inequality for total variation and \eqref{eq:EricEtAl_1} yield \eqref{eq:EricEtAl_S1}. Under $H_0$, the $2k$ pooled ideal outcomes are i.i.d.\ from a common law, so $E[\phi(S)]\le\alpha$. Combining \eqref{eq:TVbound} and \eqref{eq:EricEtAl_S1} proves \eqref{eq:EricEtAl_S2}. Finally, $k=o(n^{2/3})$ makes the approximation error vanish.
\end{proof}

\section[Auxiliary results]{Auxiliary results}\label{app:auxiliary}

\begin{lemma}\label{lem:Fr}
Under Assumption \ref{ass:g-regular}, there exists $r_{1}>0$ and $0<C_{L}<C_{H}<\infty $ such that, for all $r\in [0,r_1]$, 
\begin{equation}\label{eq:CDF_lips1}
   F(r) :=  \int_{B_r}g(x) dx \quad \text{satisfies} \quad C_{L}r^{d} \leq  F(r)  \leq C_{H}r^{d}~. 
\end{equation}
Here $F(0):=0$, and $F(r)=Q\{R<r\}=Q\{R\le r\}$ is the continuous CDF of $R$. In addition, for any $u\in (0,C_{L}r_{1}^{d}]$, the (left-continuous) quantile function of $F$, 
\begin{equation*}
    F^{-1}(u) := \inf\{r : F(r)\ge u\}, \quad \text{satisfies} \quad  F^{-1}(u)\leq ({u}/{C_{L}})^{1/d}~.
\end{equation*}
\end{lemma}

\begin{proof}
Let $r_g>0$ be such that the Lipschitz bound in Assumption~\ref{ass:g-regular}(i) holds for $x\in\mathcal X$ with $\|x-x_0\|<r_g$, and set
\[
r_1:=\min\left\{\frac{r_g}{2},\frac{r_0}{2},
\frac{g(x_0)}{2\max\{C_g,1\}}\right\}>0~.
\]
Both parts of Assumption~\ref{ass:g-regular} therefore apply for $0<r\le r_1$. Since $X$ has a Lebesgue density, every sphere has probability zero, which proves the asserted continuity and identifies $F$ as the CDF of $R$. For any $x \in \mathcal{X}$ with $\Vert x-x_{0}\Vert \leq r $ and $r\in [ 0,r_{1}]$, we have
\begin{equation}\label{eq:CDF_lips2}
    g(x) \overset{(1)}{\geq }g(x_{0}) -C_gr\overset{(2)}{ \geq }( g(x_{0}) /2) ( 2-C_g/(\max \{ C_g,1\} )) \geq g(x_{0}) /2\overset{(3)}{>}0~,
\end{equation}
where (1) and (3) hold by Assumption \ref{ass:g-regular} and (2) by $r \leq r_{1}$. Then, for all $r\in [ 0,r_{1}] $,
\begin{equation}\label{eq:Fr-lowerbound}
    F(r) \overset{(1)}{\geq}\frac{g(x_{0}) }{2} \int_{B_r \cap \mathcal X}dx \overset{(2)}{\ge} \frac{g(x_{0}) }{2} C_{x_0} r^d\overset{(3)}{=}C_{L}r^{d}~,
\end{equation}
where (1) holds by \eqref{eq:CDF_lips2} and $\Vert x-x_{0}\Vert \leq r$, (2) by Assumption \ref{ass:g-regular}, and (3) by setting $C_{L}:=\tfrac{g(x_{0})}{2}C_{x_0}$. This proves the first inequality in \eqref{eq:CDF_lips1}. 

For any $x \in \mathcal{X}$ with $\Vert x-x_{0}\Vert \leq r $ and $r\in [ 0,r_{1}]$, we also have
\begin{equation}\label{eq:CDF_lips3}
    g(x) \overset{(1)}{\leq}g(x_{0}) + C_g r\overset{(2)}{ \leq }g(x_{0}) +C_g r_{1}~,  
\end{equation}
where (1) holds by Assumption \ref{ass:g-regular} and (2) by $r\leq r_{1}$. Then, for all $r\in [0,r_{1}]$,
\begin{equation*}
    F(r) \overset{(1)}{\leq }(g(x_{0}) + C_g r_{1})\mathrm{Vol}(B_r)\overset{(2)}{=}C_{H}r^{d}~,
\end{equation*}
where (1) holds by \eqref{eq:CDF_lips3} and $\Vert x-x_{0}\Vert \leq r$, and (2) by using that $\mathrm{Vol}(B_r) =\pi ^{d/2}/\Gamma ( {d}/{2}+1) r^{d}$ and setting $C_{H}:=(g(x_{0}) +C_g r_{1}) \pi ^{d/2}/\Gamma ( {d}/{2}+1) $. This proves the second inequality in \eqref{eq:CDF_lips1}.

For the last claim, take any $r\in (0,r_{1}]$, and note that 
\begin{equation}\label{eq:CDF_lips4}
    F^{-1}(C_{L}r^{d}) \overset{(1)}{\leq} F^{-1}(F(r))  \overset{(2)}{\leq }r~,
\end{equation}
where (1) holds by the weak monotonicity of the quantile function and \eqref{eq:Fr-lowerbound}, and (2) by elementary properties of the quantile function (e.g., \citet[Eq.\ 36]{pollard:02}). For any $u\in (0,C_{L}r_{1}^{d}]$, applying \eqref{eq:CDF_lips4} to $r=(u/C_{L})^{1/d}\in(0,r_{1}]$ proves the desired result.
\end{proof}

\begin{lemma}\label{lem:Dx}
Let Assumptions~\ref{ass:g-regular} and \ref{ass:QMD} hold and define, for $x\in\mathcal X$ near $x_0$ for which $p_x$ is defined,
\[
D(x) := \int \Big( \sqrt{p_x(y)} - \sqrt{p_{x_0}(y)} \Big)^2 \,\nu(dy)~.
\]
Let
\[
I(x_0) := \int \dot\ell_{x_0}(y)\,\dot\ell_{x_0}(y)^\top\, p_{x_0}(y)\,\nu(dy)
\]
denote the $d\times d$ matrix induced by the score $\dot\ell_{x_0}$. Then, as $x\to x_0$,
\[
D(x) = \frac{1}{4}\,(x-x_0)^\top I(x_0)\,(x-x_0) + o\big(\|x-x_0\|^2\big)~.
\]
\end{lemma}

\begin{proof}
For $t=x-x_0$ with $x\in\mathcal X$ near $x_0$ and $p_x$ defined, set
\[
\begin{aligned}
a_t(y)&:=\tfrac12t^\top\dot\ell_{x_0}(y)\sqrt{p_{x_0}(y)},\qquad 
e_t(y):=\sqrt{p_{x_0+t}(y)}-\sqrt{p_{x_0}(y)}-a_t(y)~.
\end{aligned}
\]
Writing $\|\cdot\|_2$ for the norm in $L^2(\nu)$, QMD gives $\|e_t\|_2^2=o(\|t\|^2)$ and score integrability gives
\begin{align}\label{eq:Dx_3}
\begin{aligned}
\|a_t\|_2^2
&=\tfrac14t^\top I(x_0)t
=\tfrac14\int (t^{\top}\dot{\ell}_{x_0}(y))^2 p_{x_0} (y) ~\nu (dy)\\
&\leq \tfrac14 \|t\|^2 \int \|\dot{\ell}_{x_0}(y)\|^2 p_{x_0} (y) ~\nu (dy)
=\tfrac14\|t\|^2\operatorname{tr}I(x_0)~.
\end{aligned}
\end{align}
where the inequality follows from Cauchy-Schwarz. Expanding the square yields
\begin{align}\label{eq:Dx_1}
D(x)
&=\| a_t + e_t \|^2_2
= \| a_t \|^2_2 +  \| e_t \|^2_2 + 2\int a_t(y)e_t(y) \nu(dy)~.
\end{align}
Since $\|a_t\|_2^2
=\tfrac14t^\top I(x_0)t$, the triangle and Cauchy--Schwarz inequalities,
\begin{equation}\label{eq:Dx_4}
\begin{aligned}
\left|D(x)-\tfrac14t^\top I(x_0)t\right|
&\le\|e_t\|_2^2+2\|a_t\|_2\|e_t\|_2=o(\|t\|^2)~.
\end{aligned}
\end{equation}
All expressions are zero at $x=x_0$, so the decomposition holds there as well. This proves the claimed expansion.
\end{proof}

\bibliography{hellinger.bib}

\newpage 

\clearpage
\begin{titlepage}
  \centering

  {\Large \textsc{Supplementary Appendix for}\\[0.5em]
  \textbf{``Approximating an Ideal Experiment with \\ 
  Induced Order Statistics''}}\\[2em]

  {\large
  \begin{minipage}[t]{0.45\textwidth}
    \centering
    Federico A. Bugni\\
    Department of Economics\\
    Northwestern University\\
    \url{federico.bugni@northwestern.edu}
  \end{minipage}
  \hfill
  \begin{minipage}[t]{0.45\textwidth}
    \centering
    Ivan A.\ Canay\\
    Department of Economics\\
    Northwestern University\\
    \url{iacanay@northwestern.edu}
  \end{minipage}

  \par\vspace{1.5em}

  \begin{minipage}[t]{\textwidth}
    \centering
    Deborah Kim\\
    Department of Economics\\
    University of Warwick\\
    \url{deborah.kim@warwick.ac.uk}
  \end{minipage}
  \par
  }
  \vspace{3em}
    \today 
  \begin{spacing}{1.1}
  \begin{abstract}
  \noindent
  This document provides supplementary material for the paper
  ``Approximating an Ideal Experiment with Induced Order Statistics.'' 
  It collects auxiliary lemmas and technical proofs omitted
  from the main text, develops additional local approximation rate results under
  alternative smoothness assumptions, and provides further
  examples and discussion.
  \end{abstract}
  \end{spacing}

  \vspace{1em}
  \noindent\textbf{KEYWORDS:} Hellinger distance, total variation distance,
  induced order statistics, rank tests, permutation tests.

  \vspace{0.5em}
  \noindent\textbf{JEL codes:} C12, C14.

  \vfill
\end{titlepage}

\renewcommand{\thesection}{S.\arabic{section}}
\renewcommand{\thesubsection}{S.\arabic{section}.\arabic{subsection}}
\renewcommand{\theequation}{S.\arabic{section}.\arabic{equation}}
\numberwithin{equation}{section}   
\renewcommand{\thetheorem}{S.\arabic{theorem}}
\renewcommand{\thelemma}{S.\arabic{lemma}}
\renewcommand{\thecorollary}{S.\arabic{corollary}}
\renewcommand{\theproposition}{S.\arabic{proposition}}
\renewcommand{\theassumption}{S.\arabic{assumption}}
\renewcommand{\thedefinition}{S.\arabic{definition}}
\renewcommand{\theexample}{S.\arabic{example}}
\renewcommand{\theremark}{S.\arabic{remark}}
\setcounter{section}{0}
\setcounter{theorem}{0}
\setcounter{lemma}{0}
\cleardoublepage
\setcounter{section}{0} 
\setcounter{lemma}{0}
\setcounter{equation}{0}


\startcontents[appendix]   
\printcontents[appendix]{l}{1}{


\section*{Supplementary Contents}}

\section[\hspace{2mm}Supplementary Lemmas]{Supplementary Lemmas}\label{supp:more-lemmas}

\begin{lemma}\label{lem:QMD-sharpness}
For every $d,m\ge1$, there exists an admissible joint law $Q$, a covariate support $\mathcal X$, and a boundary point $x_0\in\mathcal X$ satisfying Assumptions~\ref{ass:g-regular}--\ref{ass:QMD}, together with constants $C>0$ and $\bar r>0$, such that
\begin{equation}\label{eq:QMD-sharpness_0}
\mathrm H(P_r,P)\ge Cr,
\qquad\mathrm{TV}(P_r,P)\ge Cr
\quad\text{for all }r\in(0,\bar r)~.
\end{equation}
\end{lemma}

\begin{proof}
Take $x_0=0_d$ and $X\sim\mathrm{Unif}([0,1]^d)$. Then $\mathcal X=[0,1]^d$, $g(x_0)=1$, and Assumption~\ref{ass:g-regular} holds with $C_g=0$, $r_0=1$, and $C_{x_0}=2^{-d}\pi^{d/2}/\Gamma(d/2+1)$, since $\mathrm{Vol}(B_r\cap\mathcal X)=2^{-d}\mathrm{Vol}(B_r)=C_{x_0}r^d$ for $0<r<1$. Let $Y\mid X=x\sim N((x_1,0,\ldots,0)^\top,I_m)$, with $\nu$ equal to Lebesgue measure on $\mathbf R^m$. The score with respect to the $d$-dimensional covariate is
\begin{equation*}
    \nabla_x\log p_x(y)=(y_1-x_1,0,\ldots,0)^\top\in\mathbf R^d~.
\end{equation*}
At $x_0$, take $\dot\ell_{x_0}(y)=(y_1,0,\ldots,0)^\top\in\mathbf R^d$, so that
\begin{equation*}
    \int\|\dot\ell_{x_0}(y)\|^2p_{x_0}(y)\,dy=1~.
\end{equation*}
For an admissible perturbation $t$, let $s=t_1$, $q_s=\sqrt{p_{x_0+t}}$, and $q_0=\sqrt{p_{x_0}}$. Completing the square in the Gaussian integrals gives
\[
\int q_sq_0\,dy=e^{-s^2/8},
\qquad\int y_1q_sq_0\,dy=\tfrac{s}{2}e^{-s^2/8}~.
\]
Therefore,
\begin{equation*}
\begin{aligned}
\int\left(q_s-q_0-\tfrac{s}{2}y_1q_0\right)^2dy
&=2+\tfrac{s^2}{4}-2\left(1+\tfrac{s^2}{4}\right)e^{-s^2/8}=O(s^4)=o(\|t\|^2)~.
\end{aligned}
\end{equation*}
The remainder is exactly zero when $t_1=0$. This proves Assumption~\ref{ass:QMD}.

For the lower bound, let $A=\{y:y_1>0\}$ so that $P(A)=1/2$. Let $Z=(Z_1,\ldots,Z_d)^\top$ be uniform on $B_1\cap[0,\infty)^d$. Writing $\Phi$ and $\phi$ for the standard normal CDF and density, respectively, gives the following:
\begin{align*}
\begin{aligned}
P_r(A)-P(A)
&=E\!\left[P(Y_1>0\mid X)\mid X\in B_r\right]-\tfrac12\\
&=E\!\left[\Phi(X_1)\mid X\in B_r\right]-\tfrac12\\
&=E\!\left[\Phi(rZ_1)-\tfrac12\right]\\
&\ge \phi(1)E[Z_1] r~.
\end{aligned}
\end{align*}
where the third equality holds because $X\mid X\in B_r \overset{d}{=}rZ$ for $r\in(0,1)$, and the final inequality uses $E[\Phi(rZ_1)-\tfrac12] = E[\int^{r Z_1}_0 \phi(u) du] \ge E[\int^{r Z_1}_0 \phi(1) du]$ for $r\in(0,1)$. Integrating over the truncated unit ball yields
\[
    E[Z_1]=\frac{\Gamma(d/2+1)}{\sqrt\pi\,\Gamma((d+3)/2)}>0~.
\]
The event bound and $\mathrm{TV}\le\sqrt2\mathrm H$ give \eqref{eq:QMD-sharpness_0} with $\bar r=1$ and $C=\phi(1)E[Z_1]/\sqrt2$.
\end{proof}

\begin{lemma}\label{lem:QMD-no-poly-improvement}
For every $d,m\ge1$ and $\tilde r>0$, there exists an admissible joint law $Q$ satisfying Assumptions~\ref{ass:g-regular}--\ref{ass:QMD}, with $B_r\subset\mathcal X$ for all $0<r<\tilde r$, such that, for every $\varepsilon>0$,
\begin{equation}\label{eq:no-polynomial-example1}
    \frac{\mathrm H(P_r,P)}{r^{1+\varepsilon}}\longrightarrow\infty
\qquad\text{as }r\downarrow0~,
\end{equation}
\begin{equation}\label{eq:no-polynomial-example1b}
    \frac{\mathrm{TV}(P_r,P)}{r^{1+\varepsilon}}\longrightarrow\infty
\qquad\text{as }r\downarrow0~.
\end{equation}
In particular, neither distance admits a common $O(r^{1+\varepsilon})$ bound over a class containing this model.
\end{lemma}

\begin{proof}
Fix $d,m\ge1$ and $\tilde r>0$. We construct one law that satisfies both conclusions for every $\varepsilon>0$. We define the distribution $Q$ as follows. Let $d\geq 1$, $m\geq 1$, and $x_{0}=0_d.$ Let $X\sim \mathrm{Unif}(B_{\tilde{r}})$, so that $g(x)=1[x\in B_{ \tilde{r}}]/\mathrm{Vol}(B_{\tilde{r}})$ and $\mathcal{X}=B_{\tilde{r}}$. Next, let $e_{1}:=(1,0,\ldots ,0)^{\top }\in \mathbf{R}^{m}$ and let $\delta : \mathbf{R}\to \mathbf{R}$ be defined as follows: $\delta ( u) :=u/( 1-\ln u) $ for $u\in ( 0,\exp ( -2) ) $, and $\delta(u):=0$ otherwise. Let $\{ Y|X=x\} $ be supported on $\{0_m,e_{1}\}\subset \mathbf{R}^{m}$ with the following distribution: 
\begin{equation*}
P( Y=e_{1}|X=x) =\pi ( x) :=\frac{1}{2}+\delta ( \Vert x\Vert ) \text{ and }P( Y=0_m|X=x) =1-\pi ( x) ~.
\end{equation*}

Since $X\sim \mathrm{Unif}(B_{\tilde{r}})$ and $x_{0}=0_{d}$, Assumption~\ref{ass:g-regular} holds with $C_{g}=0$, $r_{0}=\tilde{r}$, and $C_{x_{0}}=\pi ^{d/2}/\Gamma ( d/2+1) $. Also, $B_{r}\cap \mathcal{X}=B_{r}$ for all $r\in ( 0, \tilde{r}) $. Next, we show that Assumption \ref{ass:QMD} holds with $\dot{\ell}_{x_{0}}(y):= 0_d$ and $\nu $ equal to the counting measure on $\{0_{m},e_{1}\}$. First, note that $\dot{\ell}_{x_{0}}(y)=0_d$ implies $\int \Vert \dot{\ell}_{x_{0}}(y)\Vert ^{2}p_{x_{0}}(y)\nu (dy)=0<\infty $. Second, for any $\Vert t\Vert \in (0,\exp (
-2) )$, we have 
\begin{align}
& \int ( \sqrt{p_{x_{0}+t}(y)}-\sqrt{p_{x_{0}}(y)}-\frac{1}{2}t^{\top }\dot{\ell}_{x_{0}}(y)\sqrt{p_{x_{0}}(y)}) ^{2}\nu (dy)  \notag \\
& \overset{(1)}{=}( \sqrt{\tfrac{1}{2}+\delta ( \Vert t\Vert ) }-\sqrt{\tfrac{1}{2}}) ^{2}+( \sqrt{\tfrac{1}{2}-\delta ( \Vert t\Vert ) }-\sqrt{\tfrac{1}{2}}) ^{2}  \notag \\
& \overset{(2)}{=}
2-\sqrt{1+\tfrac{2\Vert t\Vert }{( 1-\ln \Vert t\Vert ) }}-\sqrt{1-\tfrac{2\Vert t\Vert }{( 1-\ln \Vert t\Vert ) }}~,
\label{eq:no-polynomial-example3}
\end{align}%
where (1) holds by $\dot{\ell}_{x_{0}}(y)=0_d$, $p_{x_{0}+t}(e_{1})=1/2+\delta ( \Vert x_{0}+t\Vert ) $, $p_{x_{0}+t}(0_m)=1/2-\delta ( \Vert x_{0}+t\Vert ) $, and $p_{x_{0}}(y)=1/2$ for $y\in \{0_m,e_{1}\}$, (2) by $\delta ( u) :=u/( 1-\ln u) $ for $u\in (
0,\exp ( -2) ) $. By a second-order Taylor expansion of $h( u) =2-\sqrt{1+2u}-\sqrt{1-2u}$ centered at $u=0$, and that $h( 0) =h^{\prime }( 0) =0$, 
\begin{equation*}
h( u) =( ( 1+2\tilde{u}) ^{-3/2}+( 1-2\tilde{ u}) ^{-3/2}) \frac{u^{2}}{2}
\end{equation*}
with $\tilde{u}\in [0,u]$. Next, note that $\Vert t\Vert /( 1-\ln \Vert t\Vert ) \in ( 0,\exp ( -2) /3) $ for $\Vert t\Vert \in (0,\exp ( -2) )$, and $h^{\prime \prime }( u) \in ( 2,2.1) $ for $u\in ( 0,\exp ( -2) /3) $. By plugging the results on \eqref{eq:no-polynomial-example3}, we conclude that 
\begin{equation*}
\int ( \sqrt{p_{x_{0}+t}(y)}-\sqrt{p_{x_{0}}(y)}-\frac{1}{2}t^{\top } \dot{\ell}_{x_{0}}(y)\sqrt{p_{x_{0}}(y)}) ^{2}\nu (dy)=O\Big( \tfrac{\Vert t\Vert ^{2}}{( 1-\ln \Vert t\Vert ) ^{2}}\Big)~.
\end{equation*}
The right-hand side expression of the previous equation is $o( \Vert t\Vert ^{2})$ as  $\Vert t\Vert \downarrow 0$, which verifies Assumption \ref{ass:QMD}. Thus $Q$ is an admissible interior model with the required neighborhood.

Let $A:=\{e_1\}$. For $0<r<\min\{\tilde r,e^{-2}\}$,
\begin{align}
\sqrt{2}\mathrm{H}( P_{r},P)  \geq \mathrm{TV}(P_{r},P)&\geq \vert P_{r}( A) -P( A) \vert   \notag \\
&\overset{(1)}{=}{\int_{\Vert x\Vert \leq r}\delta ( \Vert x\Vert ) dx}/{\mathrm{Vol}(B_{r})} \notag \\
&\overset{(2)}{=}\tfrac{d}{r^{d}}\int_{0}^{r}{v^{d}}/{(1-\ln v)}dv \notag \\
&\overset{(3)}{\geq }\tfrac{d}{r^{d}}\int_{r/2}^{r}{v^{d}}/{(1-\ln v)}dv  \notag \\
&\geq \tfrac{d}{( 1-\ln ( r/2) ) r^{d}} \int_{r/2}^{r}v^{d}dv  \notag \\
&= \tfrac{d( 1-2^{-d-1})}{d+1 } \tfrac{r}{ 1-\ln ( r/2) }
~, \label{eq:no-polynomial-example4}
\end{align}
where (1) holds because $P(Y=e_{1}\mid X=x)=1/2+\delta(\Vert x\Vert)\geq 1/2$ for all $x$, $P(Y=e_{1}\mid X=x_{0})=1/2$ since $\delta(0)=0$, $X\sim \mathrm{Unif}(B_{\tilde r})$, and $r<\tilde{r}$, (2) by change of variables to polar coordinates, and (3) by $v^{d}/( 1-\ln v) \geq 0$ for $v<r<\exp(-2)<1$. For every $\varepsilon>0$, dividing the positive lower bound in \eqref{eq:no-polynomial-example4} by $r^{1+\varepsilon}$ gives a quantity proportional to $r^{-\varepsilon}/(1-\log(r/2))$, which tends to infinity. This proves \eqref{eq:no-polynomial-example1b}. Since $\mathrm H(P_r,P)\ge\mathrm{TV}(P_r,P)/\sqrt2$, it also proves \eqref{eq:no-polynomial-example1}.
\end{proof}

\begin{samepage}
\begin{lemma}\label{lem:QMD-TV-or}
Suppose Assumptions~\ref{ass:g-regular}--\ref{ass:QMD} hold and $x_0$ lies in the interior of $\mathcal X$. Then,
\[
    \mathrm H(P_r,P)=o(r)
    \qquad\text{and}\qquad
    \mathrm{TV}(P_r,P)=o(r)~.
\]
\end{lemma}
\end{samepage}

\begin{proof}
Let $r_1$ be as in Lemma~\ref{lem:Fr}. Since $x_0$ is an interior point of $\mathcal{X}$, there exists $r_2>0$ such that $B_{r_2}\subset\mathcal{X}$. Set $\tilde r=\min\{r_1,r_2\}>0$. For $r\in(0,\tilde r)$, Lemma~\ref{lem:Fr} gives $F(r)\ge C_Lr^d>0$, and Assumption~\ref{ass:g-regular} gives $g(x)>0$ for $x\in B_r$, so $p_x$ is well defined there.

For $r\in(0,\tilde r)$, let $T$ represent the shift from $x_0$ of the variable $X$ drawn conditional on lying within distance $r$ of $x_0$. Its density is
\[
    \mu_r(t):=\frac{g(x_0+t)I\{\|t\|<r\}}{F(r)}~,
\]
so that we can write $h_r(y)=E[p_{x_0+T}(y)]$. Define the QMD remainder
\[
    u_t(y):=\sqrt{p_{x_0+t}(y)}-\sqrt{p_{x_0}(y)}
    -\tfrac12 t^\top\dot\ell_{x_0}(y)\sqrt{p_{x_0}(y)}~.
\]
Since $h_r(y)$ is the expectation of $(\sqrt{p_{x_0+T}(y)})^{2}$, we obtain a useful auxiliary function by replacing the square-root density with its QMD approximation, 
\begin{equation}\label{eq:h-tilde-r}
        \widetilde h_r(y):=
    E\!\left[\left(\sqrt{p_{x_0}(y)}
    +\tfrac12T^\top\dot\ell_{x_0}(y)\sqrt{p_{x_0}(y)}\right)^2\right]~.
\end{equation}
For $\nu$-almost every fixed $y$, the definitions of $h_r$ and $\widetilde h_r$ give
\begin{equation*}
\sqrt{h_r(y)}=\left(\int\left|\sqrt{p_{x_0+t}(y)}\right|^2\mu_r(t)\,dt\right)^{1/2}
\end{equation*}
and
\begin{equation*}
\sqrt{\widetilde h_r(y)}=\left(\int\left|\sqrt{p_{x_0}(y)}+\tfrac12 t^\top\dot\ell_{x_0}(y)\sqrt{p_{x_0}(y)}\right|^2\mu_r(t)\,dt\right)^{1/2}~.
\end{equation*}
That is, these are the $L^2(\mu_r)$ norms of the two functions inside the absolute values, viewed as functions of $t$ with $y$ held fixed. Their difference is $u_t(y)$, so the reverse triangle inequality gives
\begin{equation*}
\left|\sqrt{h_r(y)}-\sqrt{\widetilde h_r(y)}\right|
\leq\left(\int|u_t(y)|^2\mu_r(t)\,dt\right)^{1/2}
=\left(E[u_T(y)^2]\right)^{1/2}~.
\end{equation*}
Combining this inequality with $(a+b)^2\le2a^2+2b^2$ and Tonelli's theorem yields
\begin{align}
    \mathrm H^2(P_r,P)
    &\le E\!\left[\int u_T(y)^2\,\nu(dy)\right]
    +\int\left(\sqrt{\widetilde h_r(y)}-\sqrt{p_{x_0}(y)}\right)^2\nu(dy)~.
    \label{eq:interior-H-bound}
\end{align}

Next, we show that both terms on the right-hand side of \eqref{eq:interior-H-bound} are $o(r^2)$. For the first, note that Assumption~\ref{ass:QMD} implies that 
\[
    \epsilon(r) :=  \sup_{0<\|t\|<r}
    \frac{\int u_t(y)^2\,\nu(dy)}{\|t\|^2} \to 0~.
\]
Since $\|T\|<r$ almost surely and $u_0=0$, this gives
\[
    E\!\left[\int u_T(y)^2\,\nu(dy)\right]
    \le \epsilon(r)E[\|T\|^2] \le \epsilon(r)r^2 
    =o(r^2)~.
\]
For the second term, first note that
\begin{align}
    \|E[T]\|
    \overset{(1)}{=}
    \frac1{F(r)}
    \left\|\int_{\{\|t\|<r\}}t\{g(x_0+t)-g(x_0)\}\,dt\right\|\overset{(2)}{\le}
    \frac{C_g}{C_Lr^d}\int_{\{\|t\|<r\}}\|t\|^2\,dt
    =O(r^2)~,
    \label{eq:interior-H-mean}
\end{align}
where (1) uses $\int_{\{\|t\|<r\}}t\,dt=0$, and (2) follows from Assumption~\ref{ass:g-regular} and Lemma~\ref{lem:Fr}.
Next, expanding the square in \eqref{eq:h-tilde-r}  gives
\[
    \widetilde h_r(y)=p_{x_0}(y)
    \left\{1+E[T]^\top\dot\ell_{x_0}(y)
    +\tfrac14E\!\left[(T^\top\dot\ell_{x_0}(y))^2\right]\right\} = p_{x_0}(y)\{1+O(r^2)\}~,
\]
where the second equality follows from \eqref{eq:interior-H-mean} and $E[(T^\top\dot\ell_{x_0}(y))^2]\le r^2\|\dot\ell_{x_0}(y)\|^2$. That is, the term in braces is $1+O(r^2)$ for each fixed $y$. Thus, if $p_{x_0}(y)>0$,
\[
    \frac{\sqrt{\widetilde h_r(y)}-\sqrt{p_{x_0}(y)}}r
    =\frac{\widetilde h_r(y)-p_{x_0}(y)}
    {r\{\sqrt{\widetilde h_r(y)}+\sqrt{p_{x_0}(y)}\}}
    \to 0~.
\]
If $p_{x_0}(y)=0$, then $\widetilde h_r(y)=0$ and the same conclusion holds.
Next, note that $\sqrt{\widetilde h_r(y)}$ and $\sqrt{p_{x_0}(y)}$ are the $L^2(\mu_r)$ norms of $\sqrt{p_{x_0}(y)}+\tfrac12 t^\top\dot\ell_{x_0}(y)\sqrt{p_{x_0}(y)}$ and $\sqrt{p_{x_0}(y)}$, respectively, viewed as functions of $t$ with $y$ held fixed. Thus, application of the reverse triangle inequality  in $L^2(\mu_r)$ to these functions yield
\begin{align*}
    \left|\sqrt{\widetilde h_r(y)}-\sqrt{p_{x_0}(y)}\right|
    \le \tfrac12\sqrt{p_{x_0}(y)}
    \left(E[(T^\top\dot\ell_{x_0}(y))^2]\right)^{1/2}\le \tfrac12r\|\dot\ell_{x_0}(y)\|\sqrt{p_{x_0}(y)}~.
\end{align*}
Consequently,
\[
    \frac{\left|\sqrt{\widetilde h_r(y)}-\sqrt{p_{x_0}(y)}\right|^2}{r^2}
    \le \tfrac14\|\dot\ell_{x_0}(y)\|^2p_{x_0}(y)~.
\]
The right-hand side is integrable by Assumption~\ref{ass:QMD}. Dominated convergence therefore implies that the second term in \eqref{eq:interior-H-bound} is $o(r^2)$. This proves $\mathrm H(P_r,P)=o(r)$, and the total variation conclusion follows from $\mathrm{TV}(P_r,P)\le\sqrt2\,\mathrm H(P_r,P)$.
\end{proof}

\begin{lemma}\label{lem:ordered_uniform_power}
Let $U_{( k:n) }$ be the $k$-th order statistic from an i.i.d.\ sample $U_{1},\dots ,U_{n}\sim \mathrm{U}(0,1)$. For any real $m>0$,
\begin{equation*}
E[ U_{( k:n) }^{m}] =O( ({k}/{n}) ^{m}) ~.
\end{equation*}
\end{lemma}

\begin{proof}
For any $m>0$, $k\in \{1,\dots,n\}$, and $n\in \mathbf{N}$, consider the following derivation:
\begin{equation}\label{eq:ordered_1}
    E[ U_{(k:n)}^{m}] \overset{(1)}{=}\frac{B( m+k,n-k+1) }{B( k,n-k+1) }\overset{(2)}{=}\frac{\Gamma(m+k) /\Gamma (k) }{\Gamma( m+n+1)/\Gamma( n+1) }~,  
\end{equation}
where (1) holds by \citet[p.~10-11]{david/nagaraja:2003}, \citet[6.2.1, p.~257]{abramowitz/stegun:1964}, where $B( a,b)
=\int_{0}^{1}t^{a-1}( 1-t) ^{b-1}dt:( 0,\infty ) ^{2}\to ( 0,\infty ) $ is the Beta function, and (2) by \citet[6.2.2, p.~257]{abramowitz/stegun:1964}, which gives $B( a,b) =\Gamma ( a) \Gamma ( b) /\Gamma ( a+b) $, where $\Gamma ( a) =\int_{0}^{\infty }t^{a-1}\exp ( -t) dt:( 0,\infty ) \to ( 0,\infty ) $ is the Gamma function.

Next, we show that, for all $m>0$, we can find constants $C_{1}(m),C_{2}( m)$ with $0<C_{1}( m) <C_{2}( m) <\infty $ such that for all $t\in \mathbf{N}$, 
\begin{equation}\label{eq:ordered_2}
    0<C_{1}( m) \leq \frac{\Gamma ( m+t) }{t^{m}\Gamma ( t) }\leq C_{2}( m) < \infty ~.
\end{equation}

By definition, $\Gamma :( 0,\infty ) \to \mathbf{R}$ is continuous and positive. Therefore, for all $t\in \mathbf{N}$, $t^{-m}\Gamma
( m+t) /\Gamma ( t) $ is positive. Also, \citet[6.1.46, p.~257]{abramowitz/stegun:1964} implies that $\lim_{t\to \infty }t^{-m}\Gamma ( m+t) /\Gamma ( t) =1$ for any $m\in \mathbf{R} $. Therefore, $\exists t_{1}\in \mathbf{N}$ such that for all $t>t_{1}$, $t^{-m}\Gamma ( m+t) /\Gamma ( t) \in (1/2,2)$. Then, for all $t\in \mathbf{N}$, 
\begin{align*}
    \frac{\Gamma ( m+t) }{t^{m}\Gamma ( t) } 
    \leq \max \left\lbrace \max_{t\leq t_{1}}\frac{\Gamma ( m+t) }{t^{m}\Gamma ( t) },\max_{t>t_{1}}\frac{\Gamma ( m+t) }{t^{m}\Gamma ( t) }\right\rbrace  
    \leq \max \left\lbrace  \max_{t\leq t_{1}}\frac{\Gamma ( m+t) }{t^{m}\Gamma ( t) },2 \right\rbrace := C_{2}( m) < \infty
\end{align*}
and
\begin{align*}
    \frac{\Gamma ( m+t) }{t^{m}\Gamma ( t) } 
    \geq \min \left\lbrace \min_{t\leq t_{1}}\frac{\Gamma ( m+t) }{t^{m}\Gamma ( t) },\min_{t>t_{1}}\frac{\Gamma ( m+t) }{t^{m}\Gamma ( t) }\right\rbrace 
    \geq \min \left\lbrace \min_{t\leq t_{1}}\frac{\Gamma ( m+t) }{ t^{m}\Gamma ( t) }, \frac12 \right\rbrace := C_{1}( m) >0~,
\end{align*}
as desired in \eqref{eq:ordered_2}.

To conclude the proof, for any $k\in\{1,\dots,n\}$ and $n\in\mathbf N$, by \eqref{eq:ordered_1},
\[
E[U_{(k:n)}^m]
=\frac{\Gamma(m+k)/\Gamma(k)}{\Gamma(m+n+1)/\Gamma(n+1)}~.
\]
Applying \eqref{eq:ordered_2} with $t=k$ in the numerator and with $t=n+1$ in the denominator gives
\[
E[U_{(k:n)}^m] \le \frac{C_2(m)k^m}{C_1(m)(n+1)^m} =
\frac{C_2(m)}{C_1(m)}\Big(\frac{k}{n+1}\Big)^m \le \frac{C_2(m)}{C_1(m)}\Big(\frac{k}{n}\Big)^m~.
\]
Hence $E[U_{(k:n)}^m]=O((k/n)^m)$.
\end{proof}

\begin{lemma}\label{lem:FHR-Lp}
Suppose $g(x_0)>0$ and Assumption \ref{ass:book} holds with additional integrability condition: $\zeta_2\in L^p(p_{x_0}\cdot\nu)$. Then, for $p\geq1,$
\begin{align*}
    \mathrm{H}(P_r, P) = O(r^{\min\{p,2\}})\quad \text{ and }\quad\mathrm{TV}(P_r, P) = O(r^2)~.
\end{align*}
\end{lemma}

\begin{proof}
By Assumption \ref{ass:book}, there exist $\tilde\zeta_{2,v}$ and $\varepsilon_0>0, C_1>0$ uniformly for all $v\in\mathbf{R}^d$ with $\|v\|\le \varepsilon_0$ and all $y\in\mathbf{R}^m$,
\[
    f(x_0+v,y)=f(x_0,y)\big[ 1 + \zeta_1(y)^{\top}v + \|v\|^2\tilde\zeta_{2,v}(y)\big],
    \qquad |\tilde\zeta_{2,v} (y)| \leq C_1\zeta_2(y)~.
\]
Integrating with respect to $y$ gives 
\[
    g(x_0+v) = g(x_0)\big[ 1 + \bar\zeta_1^{\top}v + \bar\zeta_{2,v}\|v\|^2\big]~,
\]
where 
\[
    \bar\zeta_1 := \int \zeta_1(y) p_{x_0}(y) \nu(dy), \qquad
    \bar\zeta_{2,v} := \int \tilde\zeta_{2,v}(y) p_{x_0}(y) \nu(dy)~.
\]
Since $ \zeta_2(y) \in L^p(p_{x_0}\cdot\nu)$, $\bar\zeta_{2,v}$ is bounded uniformly for $\|v\|\leq \varepsilon_0$:
$$|\bar\zeta_{2,v}|\leq \int |\tilde\zeta_{2,v}(y)| p_{x_0}(y) \nu(dy)~\leq C_1 \int \zeta_2(y)p_{x_0}(y)\nu(dy)=: C_1\bar\zeta_2 <\infty~.$$

Choose $r \in (0, \varepsilon_0). $ To obtain $h_r(y)$, we integrate the joint density over $B_r$. Let $B_r(0)$ be a ball of radius $r$ around the origin. Since $x_0\in\mathrm{int}(\mathcal X)$, by symmetry, the numerator is 
\begin{gather*}
\int_{B_r(0)} f(x_0+v,y)\,dv
= \mathrm{Vol}(B_{r})f(x_0,y)[1+A_r(y)],
 ~~\text{where}~ A_r(y):=\frac{1}{\mathrm{Vol}(B_{r})}\int_{B_r(0)} \|v\|^2 \tilde\zeta_{2,v}(y)\,dv~.
\end{gather*}
On $B_r(0)$, $\|v\|^2 \leq r^2$, so $A_r(y)$ is bounded from above:
\begin{align}
    |A_r(y)|\le \frac{1}{\mathrm{Vol}(B_{r})}\int_{B_r(0)} \|v\|^2 |\tilde\zeta_{2,v}(y)|\,dv 
    \le  \frac{1}{\mathrm{Vol}(B_{r})}\int_{B_r(0)} r^2 C_1 \zeta_{2}(y)\,dv = C_1r^2\,\zeta_2(y)~.
\end{align}
To get the denominator, integrate $g(x_0+v)$ over $B_r(0)$. Then, by Tonelli's theorem, we have
\begin{gather*}
\int_{B_r(0)} g(x_0+v)\,dv 
= \int_{B_r(0)}\int f(x_0+v, y)\,\nu(dy)\,dv 
=\mathrm{Vol}(B_{r})\int f(x_0,y)  [1+ A_r(y)]\,\nu(dy)~.
\end{gather*}
Therefore, we have
\begin{gather*}
\int_{B_r(0)} g(x_0+v)\,dv= \mathrm{Vol}(B_{r}) g(x_0)[1+\bar{A}_r] \quad\text{where}\quad \bar{A}_r:=\frac{1}{g(x_0)}\int f(x_0, y) A_r(y)\,\nu(dy)~,
\end{gather*}
Using the upper bound $|A_r(y)|\leq C_1 r^2 \zeta_2(y)$, we have
\begin{align*}
    |\bar A_r|\le \frac{1}{g(x_0)} \int f(x_0, y) |A_r(y)|\,\nu(dy) \le C_1r^2 \int p_{x_0}(y) \zeta_2(y) \,\nu(dy) \leq C_1r^2\bar\zeta_2~.
\end{align*}
Combining the denominator and the numerator, we get
\[
h_r(y)=p_{x_0}(y)\,\frac{1+A_r(y)}{1+\bar A_r}~.\] 
Now for simplicity, define \[
\Delta_r(y):=\frac{h_r(y)}{p_{x_0}(y)}-1=\frac{A_r(y)-\bar A_r}{1+\bar A_r}~.
\]
From the upper bound of $|\bar A_r|$, there exists $\varepsilon_1\in(0, \varepsilon_0)$ such that, for any $r\in(0,\varepsilon_1)$, $|\bar A_r|\le \tfrac12$, so
\begin{align}\label{eq:book-Lp1}
    |\Delta_r(y)|
\;\le\;2\big(|A_r(y)|+|\bar A_r|\big)
\;\le\; 2C_1\,r^2\,(\zeta_2(y)+\bar\zeta_2) =:  C_2\,r^2 u(y)
\end{align}
where $C_2:= 2C_1>0$ and $u(y):=\zeta_2(y)+\bar\zeta_2$. Then, the $L^p(p_{x_0} \cdot \nu)$-integrability condition of $\zeta_2$ implies that $u \in L^p(p_{x_0} \cdot \nu)$, i.e., 
\begin{align}\label{eq:book-Lp0}
    \|u\|^p_{L^p(p_{x_0} \cdot \nu)}:=\int |\zeta_2(y)+\bar\zeta_2|^p ~p_{x_0}(y) ~\nu(dy) \leq 2^{p-1} \int (|\zeta_2(y)|^p+|\bar\zeta_2|^p ) \,p_{x_0}(y) \,\nu(dy) < \infty~,
\end{align}
as $(\tfrac{a+b}{2})^{p} \leq \tfrac12(a^p + b^p)$ holds by convexity for any $p\geq 1$ and $a,b\geq0.$

\textbf{Rate of $\mathrm{H}(P_r, P)$.} Under our assumptions, $h_0 (y):=\lim_{r\to0} h_r(y) = p_{x_0}(y)$ for all $y\in\mathcal Y$. We use the inequality: for $s>-1$, $(\sqrt{1+s}-1)^2\le \min\{s^2,|s|\}$, which gives
\begin{align}\label{eq:book-Lp2}
    \left(\sqrt{h_r(y)} - \sqrt{h_0(y)}\right)^2= \left(\sqrt{1+\Delta_r(y)}-1\right)^2h_0(y)
\leq\min\{\Delta_r^2(y),|\Delta_r(y)|\}~h_0(y)~.
\end{align}

\textit{Case $1\le p\le 2$.} For simplicity, let $\tau_r:=r^{-2}.$
Define $\mathcal{A}_r:=\{y\in\mathcal Y: u(y)\le \tau_{r}\}$ and $\mathcal{A}_r^c:=\mathcal Y\setminus \mathcal{A}_r$.
By \eqref{eq:book-Lp2}, we have
\begin{align*}
\mathrm{H}(P_r, P)^2 
&\overset{(1)}{\leq}  \frac12\int \min\{\Delta_r^2(y),|\Delta_r(y)|\} ~h_0(y)~\nu(dy)\\
&\overset{(2)}{\leq}  \frac12\int_{\mathcal{A}_r}\Delta^2_r(y) ~h_0(y) ~\nu(dy) +\frac12\int_{\mathcal{A}^c_r} |\Delta_r(y)| ~h_0(y)~\nu(dy) \\
&\overset{(3)}{\leq} \frac{C^2_2 r^4}{2} \int_{\mathcal{A}_r}u(y)^2 ~h_0(y)~\nu(dy) 
+ \frac{C_2 r^2}{2}\int_{\mathcal{A}^c_r} u(y)~h_0(y)~\nu(dy)\\
&\overset{(4)}{\leq} \frac{C^2_2 r^4}{2} \int_{\mathcal{A}_r}\tau_r^{\,2-p} u^p(y) ~h_0(y)~\nu(dy) + \frac{C_2 r^2}{2}\int_{\mathcal{A}^c_r}   \tau_r^{\,1-p} u^p(y)~h_0(y)~\nu(dy)\\
&\overset{(5)}{\leq} \frac{C^2_2 r^{2p}}{2}  \int  u^p(y)~ h_0(y)~\nu(dy) 
+ \frac{C_2 r^{2p}}{2} \int u^p(y) ~h_0(y)  ~\nu(dy)
=r^{2p} \frac{C^2_2 + C_2}{2} \|u\|^p_{L^p(p_{x_0}\cdot\nu)}~,
\end{align*}
where (1) follows from \eqref{eq:book-Lp2}, we have (2) by splitting $\mathcal Y$ into $\mathcal{A}_r$ and $\mathcal{A}^c_r$ and by the definition of minimum, (3) follows from \eqref{eq:book-Lp1}, $(4)$ follows from 
\begin{align*}
    &\frac{u(y)}{\tau_r} \leq 1 ~\to~  \left(\frac{u(y)}{\tau_r}\right)^{2-p} \leq 1  ~\to~ u(y)^2\le \tau_r^{\,2-p} u^p(y) \quad \text{ for any }y\in \mathcal{A}_r~,\\
    &\frac{u(y)}{\tau_r} > 1 ~\to~ \left(\frac{u(y)}{\tau_r}\right)^{1-p} \leq 1 ~\to~ u(y)\le \tau_r^{\,1-p} u^p(y) \quad \text{ for any }y\in \mathcal{A}^c_r~,
\end{align*}
and lastly (5) holds by replacing $\mathcal{A}_r$ and $\mathcal{A}^c_r$ with $\mathcal Y$ respectively and by substituting $\tau_r=r^{-2}$. By \eqref{eq:book-Lp0}, $\mathrm{H}(P_r, P) = O(r^{p}).$

\textit{Case $p\ge 2$.}
Similar as before, by using \eqref{eq:book-Lp2} and $(\sqrt{1+s}-1)^2\le s^2$, we have
\begin{align*}
\mathrm{H}(P_r, P)^2 
&\leq  \frac12\int \Delta_r^2(y)~h_0(y) ~\nu(dy) 
\leq \frac{C^2_2 r^4}{2} \int  u(y)^2~h_0(y) ~\nu(dy) =\frac{r^4 C^2_2}{2} \|u\|^2_{L^2(p_{x_0}\cdot\nu)}~,
\end{align*}
where the second inequality follows from \eqref{eq:book-Lp1}. Since $p\ge2$ and $u \in L^p(p_{x_0}\cdot\nu)$ as shown in \eqref{eq:book-Lp0}, the Lyapunov's inequality gives $u \in L^2(p_{x_0}\cdot\nu)$ as well. Thus, $\mathrm{H}(P_r, P)=O(r^2)$.
Combining the two cases, we conclude $\mathrm{H}(P_r, P)=O\big(r^{\min\{p,2\}}\big).$

\textbf{Rate of $\mathrm{TV}(P_r, P)$.} The total variation distance between $P_r$ and $P$ is
\begin{align*}
    \operatorname{TV}(P_r,P)
    &= \tfrac12\int |h_r(y)-h_0(y)|\,\nu(dy)
    = \tfrac12\int |\Delta_r(y)|h_0(y)\,\nu(dy)\\
    &\leq \tfrac{C_2}{2} r^2 \int u(y)h_0(y)\,\nu(dy)=\tfrac{C_2}{2} r^2 \|u\|_{L^1(p_{x_0}\cdot\nu)}~.
\end{align*}
Since $\|u\|_{L^1(p_{x_0}\cdot\nu)}<\infty$, we have $\operatorname{TV}(P_r,P) = O(r^2)$ as $r\to0.$
\end{proof}

\begin{lemma}\label{lem:book-implies-QMD}
	Assumption \ref{ass:book} implies Assumption \ref{ass:g-regular} with $x_0\in {\rm int}(\mathcal X)$ and Assumption \ref{ass:QMD}. The converse is not true. 
\end{lemma}

\begin{proof}
By Assumption \ref{ass:book}, there exist $\tilde\zeta_{2,v}$ and $C_0, \varepsilon_0>0$ uniformly for all $v\in\mathbf{R}^d$ with $\|v\|\le \varepsilon_0$ and all $y\in\mathbf{R}^m$,
\begin{equation}\label{eq:starting-expansion}
    f(x_0+v,y) =f(x_0,y)\big[ 1 + \zeta_1(y)^{\top}v+\|v\|^2\tilde\zeta_{2,v}(y)\big],
    \qquad |\tilde\zeta_{2,v} (y)| \leq C_0\zeta_2(y)~.    
\end{equation}
Integrating with respect to $y$ gives 
\begin{equation}\label{eq:g-FHR}
    g(x_0+v) = g(x_0)\big[ 1 + \bar\zeta_1^{\top}v + \bar\zeta_{2,v}\|v\|^2\big]
\end{equation}
where 
\[
    \bar\zeta_1 := \int \zeta_1(y) p_{x_0}(y) \nu(dy),\qquad
    \bar\zeta_{2,v} := \int \tilde\zeta_{2,v}(y)\,p_{x_0}(y)\,\nu(dy)~.
\]

To see Assumption \ref{ass:g-regular} holds, note that since $x_0\in {\rm int}(\mathcal X)$, $B_r \subseteq \mathcal X$ for all small $r$, and so $\mathrm{Vol}(B_r \cap \mathcal X)=\mathrm{Vol}(B_r)$ is of order $r^d$. In addition, since $ \zeta_2(y) \in L^2(p_{x_0}\cdot\nu)$, $\bar\zeta_{2,v}$ is bounded uniformly for $\|v\|\leq \varepsilon_0$:
$$|\bar\zeta_{2,v}|\leq C_0\int \zeta_2(y)p_{x_0}(y)\nu(dy)=: \bar\zeta_2 <\infty~.$$
Let $\varepsilon_1:=\min\{1,\varepsilon_0\}$. Then, for all $v$ with $\|v\|\leq \varepsilon_1$
\begin{align*}
    |g(x_0+v) - g(x_0)| 
    \overset{(a)}{=} |g(x_0) \bar\zeta_1^\top v + g(x_0)\bar \zeta_{2,v} \|v\|^2| 
    \overset{(b)}{\leq} g(x_0) \{\|\bar\zeta_1\| \cdot \|v\| + |\bar\zeta_{2,v}| \|v\|^2\}
    \overset{(c)}{\leq} C_g \|v\|
\end{align*}
where (a) follows from \eqref{eq:g-FHR}, (b) holds by Cauchy-Schwarz and the triangle inequality, and (c) holds because $\|v\|^2 \leq \|v\|$ and by setting $C_g:= g(x_0) (\|\bar\zeta_1\|+\bar\zeta_2).$

Next to show that Assumption \ref{ass:QMD} holds, let
\[
    \dot\ell_{x_0}(y) := \zeta_1(y)-\bar\zeta_1, \qquad
    a_{v}(y) := \zeta_1(y)^{\top}v + \tilde\zeta_{2,v}(y)\|v\|^2,\qquad
    b_v:=\sqrt{\frac{1}{1 + \bar\zeta_1^{\top}v + \bar\zeta_{2,v}\|v\|^2}}~.
\]
Combining \eqref{eq:starting-expansion} and \eqref{eq:g-FHR}, we obtain $\sqrt{p_{x_0+v}(y)}=\sqrt{p_{x_0}(y)}\,b_v\,\sqrt{1+a_{v}(y)}$ holds. Define
\[
    \Delta_v(y) := \frac{\sqrt{p_{x_0+v}(y)}-\sqrt{p_{x_0}(y)}}{\|v\|}
    -\frac12 \dot\ell_{x_0}(y)^{\top}\frac{v}{\|v\|}\sqrt{p_{x_0}(y)}
     = \sqrt{p_{x_0}(y)}\{ T_{1,v}(y) + T_{2,v}(y) \}~,
\]
where
\begin{align}
    T_{1,v}(y) &:=\frac{b_v}{\|v\|}\Big(\sqrt{ 1 + a_{v}(y)} - 1 - \tfrac12 a_{v}(y)\Big) \label{eq:T1v}\\
    T_{2,v}(y) &:=\frac{b_v-1}{\|v\|} + \frac12\bar\zeta_1^{\top}\frac{v}{\|v\|}
     + \frac{b_v}{2\|v\|}a_{v}(y) - \frac12\zeta_1(y)^{\top}\frac{v}{\|v\|}~.\label{eq:T2v}
\end{align}
Below, we show that the integrals of the squares of $\sqrt{p_{x_0}}~T_{2,v}$ and $\sqrt{p_{x_0}}~T_{1,v}$ are $o(1)$.

\textbf{Term $T_{2,v}$.} Applying the Taylor expansion $(1+t)^{-\tfrac12}=1-\tfrac12t+O(t^2)$ around 0 to $b_v$ gives
\[
    b_v = 1 - \tfrac12(\bar\zeta_1^{\top}v + \bar\zeta_{2,v}\|v\|^2) + O((\bar\zeta_1^{\top}v + \bar\zeta_{2,v}\|v\|^2)^2) = 1 - \tfrac12(\bar\zeta_1^{\top}v + \bar\zeta_{2,v}\|v\|^2) + O(\|v\|^2)~.
\]
The second equality holds because, for some constants $\bar C, \bar\varepsilon>0$ uniformly for $\|v\|< \min(\varepsilon_1, \bar\varepsilon)$,
\begin{equation*}
    O((\bar\zeta_1^{\top}v + \bar\zeta_{2,v}\|v\|^2)^2) \leq 2\bar C((\bar\zeta_1^{\top}v)^2 + (\bar\zeta_{2,v}\|v\|^2)^2) \leq 2 \bar C(\|\bar\zeta_1\|^2 \|v\|^2 + |\bar\zeta_{2,v}|^2\|v\|^4)=O(\|v\|^2)
\end{equation*}
where the two inequalities follow from $(a+b)^2 \leq 2(a^2+b^2)$, the Cauchy-Schwarz inequality, and the last equality holds because $\|\zeta_1\|^2<\infty$ and $|\bar\zeta_{2,v}|^2< \bar\zeta_2^2$ uniformly for $\|v\|< \min(\varepsilon_1, \bar\varepsilon)$. Thus, there exist $C_1, C_2>0$, and $\varepsilon_2\in(0,\varepsilon_1)$ such that, uniformly for all $v$ with $\|v\| <\varepsilon_2$
\begin{align}
    &\left|\frac{b_v-1}{\|v\|} + \frac12\bar\zeta_1^{\top}\frac{v}{\|v\|}\right|
    \leq \tfrac12|\bar\zeta_{2,v}| \|v\| + O(\|v\|)\leq C_1 \|v\|~, \label{eq:lem1_bv_bound1}\\
    &|b_v -1| \leq \tfrac12(\|\bar\zeta_1\| \|v\| + |\bar\zeta_{2,v}|\|v\|^2) + O(\|v\|^2) \leq C_2 \|v\|~, \label{eq:lem1_bv_bound2}\\
    &b_v^2 \leq \max\{(1+C_2 \varepsilon_2)^2, (1-C_2 \varepsilon_2)^2\} =: \bar{b}\label{eq:lem1_bv_bound3}~.
\end{align}
The term $T_{2,v}(y)^2$ is bounded uniformly for all $v$ with $\|v\| <\varepsilon_2$ as follows: 
\begin{align*}
    T_{2,v}(y) ^2
    &\leq 2\left(\frac{b_v-1}{\|v\|} + \frac12 \bar\zeta_1^{\top}\frac{v}{\|v\|}\right)^2 
    + 2\left(\frac{b_v}{2\|v\|}a_{v}(y) - \frac12 \zeta_1(y)^{\top}\frac{v}{\|v\|}\right)^2\\
    & \leq 2C_1^2\|v\|^2 + 2\left( \frac{(b_v-1)}{2\|v\|} \zeta_1(y)^{\top}v + \frac{b_v \|v\|}{2}\tilde\zeta_{2,v}(y) \right)^2\\
    & \leq 2C_1^2\|v\|^2 + (b_v-1)^2 \|\zeta_1(y)\|^2 + b_v^2\|v\|^2\tilde\zeta_{2,v}(y)^2\\
    & \leq C_3\|v\|^2 \left\{1 + \|\zeta_1(y)\|^2 + C_0^2\zeta_{2}(y)^2\right\}
\end{align*}
where the first inequality holds by $(a+b)^2\le 2a^2+2b^2$, the second inequality holds by \eqref{eq:lem1_bv_bound1} and by plugging in $a_v(y)$, the third inequality holds by $(a+b)^2\le 2a^2+2b^2$ and the Cauchy-Schwarz inequality, and the final inequality follows from \eqref{eq:lem1_bv_bound2}, \eqref{eq:lem1_bv_bound3}, setting $C_3:= \max\{2C_1^2, C_2^2,\bar{b}\}$, and using $\tilde\zeta_{2,v}(y)^2\leq C_0^2\zeta_2(y)^2$.
Therefore, uniformly for all $v$ with $\|v\| <\varepsilon_2$
\begin{align*}
    \int T_{2,v}(y)^2 p_{x_0}(y) \,\nu(dy) \leq C_3\left\{1 + \int\Big(\|\zeta_1(y)\|^2 + C_0^2\zeta_{2}(y)^2\Big) p_{x_0}(y)\nu(dy)\right\} \|v\|^2 =:C_4 \|v\|^2~.
\end{align*}

\textbf{Term $T_{1,v}$.} For $u \ge -1$ the following inequality holds:
\begin{align}\label{eq:bound-expansion-sqrt1+u}
        \Big|\sqrt{1+u} - 1 - \tfrac12 u\Big| \le C_5|u|\min\{|u|,1\}
\end{align}
for some constant $C_5>0$. We apply this with $u=a_{v}(y)$ since $a_v(y)=f(x_0+v,y)/f(x_0,y)-1>-1$ for all $y$ such that $p_{x_0}(y)>0$. Partition $\mathcal{Y}$ into
\[
    E_{1} := \big\{y\in\mathcal{Y}:|a_{v}(y)|\le \delta_{v}\big\},
    \qquad
    E_{2} :=\big\{y\in\mathcal{Y}: |a_{v}(y)|> \delta_{v}\big\},
    \qquad \text{with} \qquad 
    \delta_{v} :=\sqrt{\|v\|}~.
\]
On $E_{1}$, by applying \eqref{eq:bound-expansion-sqrt1+u} and $\min\{|a_{v}(y)|,1\}\leq |a_{v}(y)|$
\[
T_{1,v}(y)^{2} = \frac{b_{v}^{2}}{\|v\|^{2}} \Big|\sqrt{ 1 + a_{v}(y)} - 1 - \tfrac12 a_{v}(y)\Big|^{2}
      \le \frac{b_{v}^{2} C^{2}_5}{\|v\|^{2}}|a_{v}(y)|^{4} \leq C^{2}_5b_{v}^{2}\delta_{v}^{2}
      \Big|\frac{a_{v}(y)}{\|v\|}\Big|^{2}~.
\]
By the triangle inequality and Cauchy-Schwarz, 
\begin{align}\label{eq:lem1-av-bound}
    \left|\frac{a_{v}(y)}{\|v\|}\right| = \left| \zeta_{1}(y)^{\top}\frac{v}{\|v\|}
      + \tilde\zeta_{2,v}(y) \|v\|\right| 
      \leq  \|\zeta_1(y)\| + |\tilde\zeta_{2,v}(y)| \|v\| \leq \|\zeta_1(y)\| +C_0 |\zeta_{2}(y)| \|v\|~,
\end{align}
uniformly for all $v$ with $\|v\| <\varepsilon_2$. Then for all $v$ with $\|v\| <\varepsilon_2$  
\begin{align*}
    \int_{E_{1}} T_{1,v}(y)^{2} p_{x_0}(y) \nu(dy)\leq C_5^2 b_v^2 \|v\| \int 2\big(\|\zeta_1(y)\|^2 + C_0^2 |\zeta_2(y)|^2 \|v\|^2\big) p_{x_0}(y) \nu(dy)  \leq C_6 \|v\| + C_7 \|v\|^3
\end{align*}
where we use by $(a+b)^2\leq 2a^2+2b^2$ for the first inequality and the second inequality follows from \eqref{eq:lem1_bv_bound3} and setting $C_6 := 2C^2_5 \bar{b} \int \|\zeta_1(y)\|^2p_{x_0}(y)\,\nu(dy) $ and $C_7 := 2C^2_5 \bar{b} C_0^2 \int |\zeta_2(y)|^2p_{x_0}(y)\,\nu(dy)$.

Since $\min\{|a_{v}(y)|,1\}\le 1$, we have, for all $v$ with $\|v\| <\varepsilon_2$,
\[
    T_{1,v}(y)^{2} \le C^{2}_5 b_{v}^{2} \Big|\frac{a_{v}(y)}{\|v\|}\Big|^{2}  \leq 2C^{2}_5 \bar{b}\{\|\zeta_1(y)\|^2 + C_0^2|\zeta_{2}(y)|^2 \|v\|^2 \}~,
\]
where the second inequality holds by applying $(a+b)^2\le 2a^2+2b^2$ to \eqref{eq:lem1-av-bound} and by using $b_v^2\leq \bar{b}$ in \eqref{eq:lem1_bv_bound3}. Integrating this on $E_2$ yields that uniformly for all $v$ with $\|v\| <\varepsilon_2$,
\[
    \int_{E_{2}} T_{1,v}(y)^{2} p_{x_0}(y) \,\nu(dy) \le 2C^{2}_5 \bar{b}
    \int_{E_{2}}\|\zeta_1(y)\|^2 p_{x_0}(y) \,\nu(dy) + C_8 \|v\|^2~,
\]
where $C_8 := 2C^{2}_5 \bar{b} C_0^2 \int |\zeta_2(y)|^2 p_{x_0}(y) \, \nu(dy)$. We finally claim that 
\begin{align}\label{eq:lem1-dct}
    \int_{E_{2}}\|\zeta_1(y)\|^2 p_{x_0}(y)\, \nu(dy)=o(1)~.
\end{align}
To see this, fix $y$ with $p_{x_0}(y)>0$. It holds from \eqref{eq:lem1-av-bound} that $|a_v(y)|/\delta_v \leq \|\zeta_1(y)\|\|v\|^{1/2} + C_0|\zeta_2(y)| \|v\|^{3/2} \to 0$ as $\|v\|\to 0$. Therefore, $1[y \in E_2] \to 0$. By dominated convergence, \eqref{eq:lem1-dct} holds where we used the integrability of $\|\zeta_1\|^2$. Therefore, for an arbitrary $\eta>0$, there exists $\varepsilon_3 \in (0, \varepsilon_2)$ such that uniformly for all $v$ with $\|v\| <\varepsilon_3$,
\[
    \int_{E_{2}} T_{1,v}(y)^{2} p_{x_0}(y)\nu(dy) \le \eta + C_8 \|v\|^2~.
\]
Combining the two regions, we have, uniformly for all $v$ with $\|v\| <\varepsilon_3$,
\[
      \int T_{1,v}(y)^{2} p_{x_0}(y)\nu(dy) \leq C_6 \|v\| + C_7 \|v\|^3 + C_8 \|v\|^2 + \eta~.
\]

To summarize, it holds that, uniformly for all $v$ with $\|v\| <\varepsilon_3$,
\begin{align*}
    \int \Delta_{v}(y)^2 \nu(dy) \leq 2\int T_{1,v}^2(y)p_{x_0}(y)\nu(dy) + 2\int T_{2,v}^2(y)p_{x_0}(y)\nu(dy)\leq 2C_6 \|v\| + C_9 \|v\|^2 + 2C_7\|v\|^3 + 2\eta~.
\end{align*}
where $C_9 := 2(C_4 + C_8)$. As $\eta\to 0$, the upper bound converges to 0 as $\|v\|\to 0$, so quadratic mean differentiability holds. Since
\begin{align*}
    \int \|\dot\ell_{x_0}(y)\|^2 p_{x_0}(y) \nu(dy) \leq 2 \int \|\zeta_1(y)\|^2 p_{x_0}(y) \nu(dy) +  2\|\bar\zeta_1\|^2 <\infty~,
\end{align*}
the integrability condition also holds. Therefore, Assumption \ref{ass:book} implies \ref{ass:QMD}.

Finally, unlike Assumption~\ref{ass:QMD}, Assumption~\ref{ass:book} imposes a support restriction. If $f(x_0,y_\ast)=0$ for some $y_\ast$, then $f(x_0+\varepsilon,y_\ast)= 0$ for all $\|\varepsilon\|\le\varepsilon_0$. Thus, the converse implication does not hold in general. To provide a simple example, let $d=1$, $x_0=0$, and $y\in\{0,1\}$, and let $\nu$ be the counting measure on $\{0,1\}$. For $x\in\mathbf R$ with $|x|$ small, set
\[
    p_x(1) := |x|^3, \qquad p_x(0) := 1-|x|^3~.
\]
This is a valid pmf (nonnegative and summing to one) for small $|x|$. To see that Assumption \ref{ass:QMD} is satisfied, take $\dot\ell_{x_0}(y):= 0\in\mathbf R$. Then, writing the $L^2(\nu)$ expression as a sum,
\[
    \sum_{y\in\{0,1\}} \left[\frac{\sqrt{p_x(y)} - \sqrt{p_{x_0}(y)}}{|x|} - \tfrac12 \dot\ell_{x_0}(y)\sqrt{p_{x_0}(y)}\right]^2 = \left(\frac{\sqrt{1-|x|^3}-1}{|x|}\right)^2 + \left(\frac{|x|^{3/2}}{|x|}\right)^2~.
\]
Using $\sqrt{1-t} = 1 - \tfrac12 t + O(t^2)$ as $t\to 0$ with $t=|x|^3$,
\[
    \left(\frac{\sqrt{1-|x|^3}-1}{|x|}\right)^2 = \left(\frac{-\tfrac12|x|^3 + O(|x|^6)}{|x|}\right)^2 = O(|x|^4)\to 0~, \qquad \left(\frac{|x|^{3/2}}{|x|}\right)^2 = |x|\to 0~.
\]
Hence the QMD condition holds. Moreover, $\int \|\dot\ell_{x_0}(y)\|^2\,p_{x_0}(y)\,\nu(dy)=0<\infty$.
However, at $y=1$ we have $f(x_0,1) = p_{x_0}(1)=0$, but for every small $x\neq 0$, $f(x,1) = p_x(1)=|x|^3>0$, which is a violation of Assumption \ref{ass:book}. 
\end{proof}

\section[\hspace{2mm}Results under Alternative Assumptions]{Results under Alternative Assumptions}\label{supp:results-Holder}

The main results of the paper are derived under differentiability in quadratic mean (QMD); see Assumption~\ref{ass:QMD}. In this appendix, we study an alternative, perhaps more standard differentiability condition: a Taylor expansion of the joint density $f(\cdot,y)$ around $x_0$ with a H\"older remainder (Assumption~\ref{ass:smoothness-holder}), which allows for a wide range of smoothness levels and remainder controls. This condition may appear similar in spirit to the multiplicative expansion in Assumption~\ref{ass:book}, since the latter is closely related to the case $(\kappa_{\rm s},\kappa_{\rm r})=(1,1)$. We will show, however, that Assumption~\ref{ass:book} is much stronger: beyond smoothness and integrability, it imposes a restrictive multiplicative structure (and associated local support restrictions) that is not required under QMD or Assumption~\ref{ass:smoothness-holder}. We begin by stating the smoothness with H\"older remainder condition that underpins the results in this section.

\begin{assumption}[Smoothness with H\"older remainder at $x_0$]\label{ass:smoothness-holder}
    Fix $\kappa_{\rm s}\in\mathbf N_0$ and $\kappa_{\rm r}\in(0,1]$. There exists $\tilde r>0$ such that, for every $y\in\mathbf R^m$ and every $t$ with $\|t\|\le \tilde r$ and $x_0+t\in\mathcal X$,
    \begin{equation}
        f(x_0+t,y) = \sum_{\|\beta\|_1 \le \kappa_{\rm s}} \frac{\partial_x^\beta f(x_0,y)}{\beta!}\,t^\beta
          + \zeta^{\rm rem}_{\kappa_{\rm s}+\kappa_{\rm r}}(t,y),
        \qquad |\zeta^{\rm rem}_{\kappa_{\rm s}+\kappa_{\rm r}}(t,y)| \le M(y)\,\|t\|^{\kappa_{\rm s}+\kappa_{\rm r}}~.
    \end{equation}
    Here $\beta=(\beta_1,\dots,\beta_d)$ is a multi-index,
    \[
        \partial_x^\beta f(x_0,y) := \frac{\partial^{\|\beta\|_1} f}{\partial x_1^{\beta_1}\cdots\partial x_d^{\beta_d}}(x_0,y),
        \qquad
        t^\beta := \prod_{j=1}^d t_j^{\beta_j},
        \qquad
        \beta! := \prod_{j=1}^d \beta_j!~.
    \]
    Moreover,
    \begin{equation}\label{eq:M-L1integral}
        \int M(y)\,\nu(dy)<\infty,
    \qquad
    \int \bigl|\partial_x^\beta f(x_0,y)\bigr| \nu(dy)<\infty
    \quad\text{for all}\quad \|\beta\|_1\le \kappa_{\rm s}~.
    \end{equation}
\end{assumption}

Assumption \ref{ass:smoothness-holder} has several ingredients. First, it requires that the joint density $f(\cdot,y)$ admits a Taylor series expansion of order $\kappa_{\rm s}$ around $x_0$, with a remainder that is controlled by H\"older continuity of order $\kappa_{\rm r}$. Second, the bound on the remainder term through the envelope $M(y)$ allows the quality of approximation to depend on $y$, but ensures uniform control in $t$ once $M(y)$ is integrable. Finally, the integrability conditions guarantee that both the derivatives and the envelope are well behaved relative to the dominating measure $\nu$, so that expansions can be integrated and used in distributional approximations.  
Together, these conditions provide a flexible smoothness requirement that interpolates between differentiability of various orders and H\"older continuity. 

Several features of Assumption~\ref{ass:smoothness-holder} are worth highlighting.

\medskip
\noindent\emph{Monotonicity in $\kappa_{\rm s}$.}
For any fixed $\kappa_{\rm r}\in(0,1]$, Assumption~\ref{ass:smoothness-holder} becomes stronger as $\kappa_{\rm s}$ increases. In particular, if it holds with $(\kappa_{\rm s}+1,\kappa_{\rm r})$, then it automatically holds with $(\kappa_{\rm s},\kappa_{\rm r})$.

\medskip
\noindent\emph{The case $\kappa_{\rm s}=0$.}
When $\kappa_{\rm s}=0$, the condition reduces to H\"older continuity,
\[
  |f(x,y)-f(x_0,y)| \;\le\; M(y)\,\|x-x_0\|^{\kappa_{\rm r}}~,
\]
with $\kappa_{\rm r}\in(0,1]$. In particular, $\kappa_{\rm r}=1$ yields (local) Lipschitz continuity.

\medskip
\noindent\emph{The case $\kappa_{\rm s}=1$.}
When $\kappa_{\rm s}=1$, Assumption~\ref{ass:smoothness-holder} takes the form
\[
    \big|\,f(x_0+t,y)- f(x_0,y) - \nabla_x f(x_0,y)^\top t \big|
    \;\le\; M(y)\,\|t\|^{1+\kappa_{\rm r}}~,
\]
i.e., pointwise differentiability of $f$ in $x$ at $x_0$ with a H\"older remainder of order $\kappa_{\rm r}$. Heuristically, while Assumption~\ref{ass:QMD} imposes differentiability in quadratic mean of $x\mapsto \sqrt{p_x}$ at $x_0$, Assumption~\ref{ass:smoothness-holder} requires pointwise differentiability of the \emph{joint} density. Neither assumption implies the other.

The main result of this section is as follows. 

\begin{theorem}\label{thm:Holder-TV-H}
Let Assumptions \ref{ass:g-regular} and \ref{ass:smoothness-holder} hold. Then,
\begin{align*}
\mathrm{TV}(P_{r},P) =O( r^{\min \{ 1,\kappa _{s}+\kappa _{r}\} }) \qquad\text{and}\qquad \mathrm{H}(P_{r},P) =O( r^{\min \{ 1/2,(\kappa _{s}+\kappa _{r})/2\} })~.
\end{align*}
Moreover, if $x_0$ is an interior point of $\mathcal X$,
\begin{align*}
\mathrm{TV}(P_{r},P) =O( r^{\min \{ 2,\kappa _{s}+\kappa _{r}\} }) \qquad\text{and}\qquad \mathrm{H}(P_{r},P) =O( r^{\min \{ 1,(\kappa _{s}+\kappa _{r})/2\} }) ~.
\end{align*}
\end{theorem}

Theorem~\ref{thm:Holder-TV-H} provides the local approximation rates implied by Assumption~\ref{ass:smoothness-holder}. These rates contrast with those obtained under the QMD condition in Theorem~\ref{thm:H-QMD}: the two sets of primitives generate distinct local exponents $(a_h,a_{tv})$, which then determine the joint rates through Theorem~\ref{thm:high-level}. Under H\"older smoothness, the exponents depend on the combined smoothness order $\kappa_{\rm s}+\kappa_{\rm r}$ \emph{and} on whether $x_0$ is an interior point of $\mathcal X$. In general,
\begin{align}
    a_h=\min\{1/2,(\kappa_{\rm s}+\kappa_{\rm r})/2\} \qquad\text{and}\qquad a_{tv}=\min\{1,\kappa_{\rm s}+\kappa_{\rm r}\}~,
    \label{eq:Holder-TV-H_1}
\end{align}
and if $x_0$ is an interior point of $\mathcal X$,
\begin{align}
    a_h=\min\{1,(\kappa_{\rm s}+\kappa_{\rm r})/2\} \qquad\text{and}\qquad a_{tv}=\min\{2,\kappa_{\rm s}+\kappa_{\rm r}\}~.
    \label{eq:Holder-TV-H_2}
\end{align}
In both \eqref{eq:Holder-TV-H_1} and \eqref{eq:Holder-TV-H_2}, one has $a_{tv}=2a_h$, so TV converges (at most) twice as fast as Hellinger.

The effect of whether $x_0$ is an interior point of $\mathcal X$ or not is analogous to that in kernel regression estimation under twice continuous differentiability. When $x_0$ lies in the interior of $\mathcal X$, the first moment integrates to zero by symmetry and the second term of the Taylor series expansion becomes the leading term. The same phenomenon explains the differences between \eqref{eq:Holder-TV-H_1} and \eqref{eq:Holder-TV-H_2}. As $\kappa_{\rm s}+\kappa_{\rm r}$ increases, the Hellinger rate saturates at $r$ (interior case) or at $r^{1/2}$ (non-interior case), while the TV rate saturates at $r^2$ (interior case) or at $r$ (non-interior case). Under QMD, the guaranteed polynomial exponent is one for both distances at interior and boundary points. At interior points, both distances are $o(r)$, but no larger polynomial exponent is guaranteed by these assumptions.

\begin{remark}
    The rates in Theorem~\ref{thm:Holder-TV-H} can be shown to be sharp by characterizing DGPs that satisfy the conditions of the theorem and yield the rates described in each case. This sharpness result thus shows that these local approximation rates cannot be improved under Assumption~\ref{ass:smoothness-holder}. We omit the details here for brevity. 
\end{remark}

Combining Theorem~\ref{thm:Holder-TV-H} with Theorem~\ref{thm:high-level} in the case with $\kappa_{\rm s}+\kappa_{\rm r}=2$ and $x_0$ belongs  in the interior of $\mathcal X$ yields
\begin{equation}\label{eq:rates-joint-k2}
    \mathrm{H}\big(\mathcal L(S_n),\mathcal L(S)\big) = O\big(k^{1/2}(k/n)^{1/d}\big),
    \qquad
    \mathrm{TV}\big(\mathcal L(S_n),\mathcal L(S)\big) = O\big(k(k/n)^{2/d}\big)~.
\end{equation}
Here $a_{tv}=2$ and $a_h=1$, so the direct bound $k(k/n)^{a_{tv}/d}$ from Theorem~\ref{thm:high-level} determines the TV rate, rather than the indirect inequality $\mathrm{TV}\le\sqrt{2}\,\mathrm{H}$. As in the local approximation case, the TV distance converges twice as fast as the Hellinger distance. Moreover, both expressions imply the same growth condition on $k$ for joint convergence:
\[
    k(k/n)^{2/d} \to 0     \qquad\Longleftrightarrow\qquad  k = o\big(n^{2/(2+d)}\big)~.
\]
For $d=1$ this becomes $k = o(n^{2/3})$, matching the threshold obtained in Section~\ref{sec:qmd}. If instead $\kappa_{\rm s}+\kappa_{\rm r}=2$ and $x_0$ does not belong in the interior of $\mathcal X$, the local exponents are $a_h=1/2$ and $a_{tv}=1$, and Theorem~\ref{thm:high-level} yields the slower joint rates
\[
    \mathrm{H}\big(\mathcal L(S_n),\mathcal L(S)\big) = O\big(k^{1/2}(k/n)^{1/(2d)}\big),
    \qquad
    \mathrm{TV}\big(\mathcal L(S_n),\mathcal L(S)\big) = O\big(k(k/n)^{1/d}\big)~,
\]
which require $k=o\big(n^{1/(1+d)}\big)$. More generally, as $\kappa_{\rm s}+\kappa_{\rm r}$ decreases, the admissible growth of $k$ slows accordingly, and the slowdown is more pronounced in the non-interior case where the linear term does not cancel.

\begin{remark}
The $L^1$ integrability requirement $\int M(y)\,\nu(dy)<\infty$ in Assumption~\ref{ass:smoothness-holder} may be strengthened to an $L^p$ condition for some $p>1$. Such a strengthening does not, by itself, improve the convergence rates for $\mathrm{H}(P_r,P)$ or $\mathrm{TV}(P_r,P)$ obtained in Theorem~\ref{thm:Holder-TV-H}. This parallels the result in Lemma~\ref{lem:FHR-Lp}. 
\end{remark}

\subsection{Auxiliary Results under Alternative Assumptions}\label{sec:auxiliary-results-holder}

Assumption~\ref{ass:book} in \cite{falk/husler/reiss:2010} implies the local approximation error $\mathrm{H}(P_r,P)=O(r^{2})$, a rate that is substantially faster than the feasible rates delivered by Theorem~\ref{thm:Holder-TV-H}, which are sharp. Because Assumption~\ref{ass:book} is closely related to Assumption~\ref{ass:smoothness-holder} with $\kappa_{\rm s}=\kappa_{\rm r}=1$, it is natural to examine more carefully which aspects of these conditions drive such a discrepancy in rates. To keep the discussion as simple as possible, in this section we focus on the univariate case $d=m=1$ and $x_0$ in the interior of the support of $X$. 

The main takeaway is that Assumption \ref{ass:book} imposes a multiplicative structure on the local approximation to $f(x,y)$ (or $p_x(y)$) that automatically restricts the behavior of $f(x,y)$ (or $p_x(y)$) in a neighborhood of the zero–density region
\begin{equation*}
    \mathcal N := \{\,y \in \mathcal Y : f(x_0,y)=0\,\}~.
\end{equation*}
Using similar arguments to those in the proof of Theorem \ref{thm:Holder-TV-H}, it is possible to show that if Assumption~\ref{ass:smoothness-holder} holds with $\kappa_{\rm s}=\kappa_{\rm r}=d=1$ and $M(y)/f(x_0,y)\in L^{2}(p_{x_0}\!\cdot\nu)$, then
\begin{equation}\label{eq:our-rate-LP}
    \mathrm{H}(P_r,P)=
    \begin{cases}
        O\!\left(r^{2}\right) & \text{if }\nu(\mathcal N)=0~,\\[3pt]
        O(r) & \text{otherwise}~,
    \end{cases}
\end{equation}
so that the baseline linear rate improves to the quadratic rate when the integrability condition in Assumption~\ref{ass:smoothness-holder} is strengthened and $\nu(\mathcal N)=0$ is imposed. These conditions are sufficient but not necessary. More precisely, under the strengthened integrability condition, $\mathrm{H}^2(P_r,P)\le Cr^4+\tfrac12P_r(\mathcal N)$, so any condition ensuring $P_r(\mathcal N)=O(r^4)$, such as $\nu(\mathcal N)=0$ or $f(x,y)=0$ for all $y\in\mathcal N$ and $x$ near $x_0$, yields the quadratic rate. Without such control of the probability entering $\mathcal N$, these assumptions do not guarantee an improvement over the baseline rate. Effectively, these additional assumptions close the gap between Assumption~\ref{ass:book} and Assumption~\ref{ass:smoothness-holder}, yielding the same fast rate. This highlights the central role of the zero–density region $\mathcal N$ in determining whether higher integrability translates into faster convergence. The details behind the result in \eqref{eq:our-rate-LP} are omitted for brevity.

\subsection{Proof of Theorem \ref{thm:Holder-TV-H}}

By Assumption \ref{ass:g-regular}(ii), for any $r\in ( 0,r_{0}) $, 
\begin{equation}
\mathrm{Vol}( B_{r}\cap \mathcal{X}) =\int_{B_{r}\cap \mathcal{X} }dx\geq C_{x_{0}}r^{d}>0~.  \label{eq:thm4_1}
\end{equation}
Let $r_{\rm s}$ denote the constant $\tilde r$ in Assumption~\ref{ass:smoothness-holder}, and let $r_1$ be as in Lemma~\ref{lem:Fr}. For the rest of the proof, let $r_{2}= \min \{ r_{0},r_{\rm s},r_1\} >0$.

To get the result for $\mathrm{TV}(P_{r},P)$, we derive expansions for any $r\in ( 0,r_{2}) $. Let
\[
\mathcal I_{\kappa_s}^{(1)}:=\big\{\beta\in \mathbf N_0^{d}: 1\le \|\beta\|_1\le \kappa_s\big\},
\qquad
\mathcal I_{\kappa_s}^{(2)}:=\big\{\beta\in \mathbf N_0^{d}: 2\le \|\beta\|_1\le \kappa_s\big\}~
\]
where $\mathbf{N}_0$ is the set of nonnegative integers.

For $\beta\in\mathcal I_{\kappa_s}^{(1)}$, define
\[
c_\beta:=\int\partial_x^\beta f(x_0,y)\,\nu(dy)~,
\qquad\text{so that}\qquad
|c_\beta|\le\int|\partial_x^\beta f(x_0,y)|\,\nu(dy)<\infty
\]
by \eqref{eq:M-L1integral}.
First, we have
\begin{align}
\int_{B_{r}\cap \mathcal{X}}f(x,y)dx &\overset{(1)}{=}\int_{B_{r}\cap \mathcal{X}}\Big( f(x_{0},y)+\sum_{\beta \in \mathcal I_{\kappa_s}^{(1)}}\frac{ \partial _{x}^{\beta }f(x_{0},y)}{\beta !} \prod_{j=1}^{d}( x_{j}-x_{0,j}) ^{\beta _{j}}
+\zeta^{\rm rem} _{\kappa _{s}+\kappa _{r}}( x-x_{0},y) \Big) dx  \notag \\
&\overset{(2)}{=}\mathrm{Vol}( B_{r}\cap \mathcal{X}) f(x_{0},y)+\sum_{\beta \in \mathcal I_{\kappa_s}^{(1)}}\frac{ \partial _{x}^{\beta }f(x_{0},y)}{\beta !}J_{\beta }( r) +R( r,y) ~, \label{eq:thm4_2} 
\end{align}
where (1) holds by Assumption \ref{ass:smoothness-holder}, and (2) by setting
\begin{align*}
J_{\beta }( r)  & := \int_{B_{r}\cap \mathcal{X} }\prod_{j=1}^{d}( x_{j}-x_{0,j}) ^{\beta _{j}}dx \\
R( r,y)  & := \int_{B_{r}\cap \mathcal{X}}\zeta^{\rm rem} _{\kappa _{s}+\kappa _{r}}( x-x_{0},y) dx~.
\end{align*}

Second, for any $r\in ( 0,r_{2}) $, we have
\begin{align}
&\int_{B_{r}\cap \mathcal{X}}g( x) dx 
=\int_{B_{r}\cap \mathcal{X}}\int f(x,y)\nu( dy) dx  \notag \\
&\overset{(1)}{=}\int_{B_{r}\cap \mathcal{X}}\Big( g(x_{0})+\sum_{\beta \in \mathcal I_{\kappa_s}^{(1)}}\frac{c_{\beta}}{ \beta !}\prod_{j=1}^{d}( x_{j}-x_{0,j}) ^{\beta _{j}}+\int \zeta^{\rm rem} _{\kappa _{s}+\kappa _{r}}( x-x_{0},y) \nu( dy) \Big) dx  \notag \\
&\overset{(2)}{=}\mathrm{Vol}( B_{r}\cap \mathcal{X}) g(x_{0})+\sum_{\beta \in \mathcal I_{\kappa_s}^{(1)}}\frac{c_\beta}{ \beta !}J_{\beta }( r) +R( r) ~,\label{eq:thm4_3}
\end{align}
where (1) holds by integrating the expansion in Assumption~\ref{ass:smoothness-holder} with respect to $\nu(dy)$, which \eqref{eq:M-L1integral} justifies, and by $g( x_{0}) =\int f(x_{0},y)\nu( dy) $, and (2) by setting
\begin{equation*}
R( r)  := \int_{B_{r}\cap \mathcal{X}}\int \zeta^{\rm rem} _{\kappa _{s}+\kappa _{r}}( x-x_{0},y) ~\nu( dy) dx~.
\end{equation*}

Finally, for any $r\in ( 0,r_{2}) $, we have
\begin{align}
h_{r}( y) - h_{0}( y)  &=\frac{\int_{B_{r}\cap \mathcal{X}}f(x,y)dx-h_{0}( y) \int_{B_{r}\cap \mathcal{X} }g( x) dx}{\int_{B_{r}\cap \mathcal{X}}g( x) dx} \notag \\
&\overset{(1)}{=}\frac{( \int_{B_{r}\cap \mathcal{X} }f(x,y)dx-h_{0}( y) \int_{B_{r}\cap \mathcal{X}}g( x) dx) /\mathrm{Vol}( B_{r}\cap \mathcal{X}) }{\int_{B_{r}\cap \mathcal{X}}g( x) dx/\mathrm{Vol}( B_{r}\cap \mathcal{X} ) }  \notag \\
&\overset{(2)}{=}\frac{ 
{\sum}_{\beta \in \mathcal I_{\kappa_s}^{(1)}}\frac{ \partial _{x}^{\beta }f(x_{0},y)-h_{0}( y) c_{\beta} }{ \beta !}\frac{J_{\beta }( r) }{\mathrm{Vol} ( B_{r}\cap \mathcal{X}) } +\frac{R(r,y) -h_{0}( y) R(r) }{\mathrm{ Vol}( B_{r}\cap \mathcal{X}) }
 }{\int_{B_{r}\cap \mathcal{X}}g( x) dx/\mathrm{Vol}( B_{r}\cap \mathcal{X})}~,\label{eq:thm4_4}
\end{align}
where (1) holds by \eqref{eq:thm4_1}, which yields $\mathrm{Vol}( B_{r}\cap \mathcal{X}) >0$, and (2) by \eqref{eq:thm4_2}, \eqref{eq:thm4_3}, and $f(x_{0},y)=h_{0}( y) g(x_{0})$.

For any $x\in \mathcal{X}$ with $\Vert x-x_{0}\Vert \leq r$ and $ r\in ( 0,r_{2}) $, \eqref{eq:CDF_lips2} in the proof of Lemma~\ref{lem:Fr} and $r_2\le r_1$ imply $g( x) \geq g( x_{0}) /2>0.$ Then, for any $r\in ( 0,r_{2}) $,
\begin{equation*}
\int_{B_{r}\cap \mathcal{X}}g( x) dx\overset{(1)}{\geq }( g( x_{0}) /2) \mathrm{Vol}( B_{r}\cap \mathcal{X} ) >0~.
\end{equation*}
Then, for all $r\in ( 0,r_{2}) $, we get 
\begin{equation}
\frac{\int_{B_{r}\cap \mathcal{X}}g( x) dx}{\mathrm{Vol}(B_{r}\cap \mathcal{X}) }\geq g( x_{0}) /2>0~.
\label{eq:thm4_5}
\end{equation}

Now there are two cases, depending on whether $x_0$ is in the interior of $\mathcal X$ or not. First, consider the case where $x_0$ is in the interior of $\mathcal X$. By definition, this implies that there exists $r_3>0$ such that $B_r \subset \mathcal X$ for all $r \in (0,r_3)$. Then, for any $r\in ( 0,\min\{r_{2},r_3\})$, we obtain 
\begin{equation}
h_{r}( y) -h_{0}( y) \overset{(1)}{=}\frac{
    \sum_{\beta \in \mathcal I_{\kappa_s}^{(2)}}\frac{\partial _{x}^{\beta }f(x_{0},y)-h_{0}( y) c_\beta }{\beta !}\frac{J_{\beta }( r) }{\mathrm{Vol}( B_{r}\cap \mathcal{X}) } +\frac{R(r,y) -h_{0}( y) R( r) }{\mathrm{ Vol}( B_{r}\cap \mathcal{X}) }}{\int_{B_{r}\cap \mathcal{X}}g( x) dx/\mathrm{Vol}( B_{r}\cap \mathcal{X}) }~,  \label{eq:thm4_6}
\end{equation}
where (1) holds by \eqref{eq:thm4_4} and $B_r \subset \mathcal X$ for all $r\in ( 0,r_3) $, and so $J_{\beta }( r)  = \int_{B_{r}}\prod_{j=1}^{d}( x_{j}-x_{0,j}) ^{\beta _{j}}dx =0$ for any $\beta \in \{ \{ 0,1\} ^{d}:\Vert \beta \Vert _{1}=1\} $. Then, for any $r\in ( 0,\min\{r_{2},r_3\}) $,  
\begin{align*}
\mathrm{TV}(P_{r},P) &=\frac{1}{2}\int \vert h_{r}( y) -h_{0}( y) \vert \nu( dy)  \\
&\overset{(1)}{\leq }\frac{ \sum_{\beta \in \mathcal I_{\kappa_s}^{(2)}}\frac{2\int \vert \partial _{x}^{\beta }f(x_{0},y)\vert \nu( dy) }{ \beta !}\frac{\vert J_{\beta }( r) \vert }{\mathrm{Vol}( B_{r}\cap \mathcal{X}) } +\frac{2\int \vert R(r,y) \vert \nu( dy) }{ \mathrm{Vol}( B_{r}\cap \mathcal{X}) } }{g( x_{0}) } \\
&\overset{(2)}{\leq }\sum_{\beta \in \mathcal I_{\kappa_s}^{(2)}}O(r^{\Vert \beta \Vert _{1}})+O(r^{\kappa _{s}+\kappa _{r}} )\\
&=O( r^{\min \{ 2,\kappa _{s}+\kappa _{r}\} }) ~,
\end{align*}
where (1) holds by \eqref{eq:thm4_6}, the triangle inequality, \eqref{eq:thm4_5}, $\int h_{0}( y) \nu( dy) =1$, $h_{0}( y) \geq 0$, $\vert c_\beta\vert \leq \int \vert \partial _{x}^{\beta }f(x_{0},y)\vert \nu( dy) $, and $\vert R(r) \vert \leq \int
\vert R(r,y) \vert \nu( dy) $, and (2) by Assumption \ref{ass:smoothness-holder}, 
\begin{align*}
\vert J_{\beta }( r) \vert =\Big\vert \int_{B_{r}\cap \mathcal{X}}\prod_{j=1}^{d}( x_{j}-x_{0,j}) ^{\beta _{j}}dx\Big\vert \leq r^{\Vert \beta \Vert _{1}}\mathrm{Vol} ( B_{r}\cap \mathcal{X})~,
\end{align*}
and
\begin{align*}
\int \vert R(r,y) \vert \nu( dy) & \leq \int M(y)\,\nu(dy)\; r^{\kappa _{s}+\kappa _{r}}\mathrm{Vol}( B_{r}\cap \mathcal{X})~ .
\end{align*}

Next, consider the case where $x_0$ is not in the interior of $\mathcal X$. The sole difference is that \eqref{eq:thm4_6} includes the term $J_{\beta }( r) $ with $\beta \in \{ \{ 0,1\} ^{d}:\Vert \beta \Vert _{1}=1\} $, which now becomes
\begin{equation}
h_{r}( y) -h_{0}( y) \overset{(1)}{=}\frac{    \sum_{\beta \in \mathcal I_{\kappa_s}^{(1)}}\frac{\partial _{x}^{\beta }f(x_{0},y)-h_{0}( y) c_{\beta} }{\prod_{s=1}^{d}\beta !}\frac{J_{\beta }( r) }{\mathrm{Vol}( B_{r}\cap  \mathcal{X}) } +\frac{R(r,y) -h_{0}( y) R(r) }{\mathrm{ Vol}( B_{r}\cap \mathcal{X}) } }{\int_{B_{r}\cap \mathcal{X}}g( x) dx/\mathrm{Vol}( B_{r}\cap \mathcal{X}) }~.
\end{equation}
By using this expression and repeating previous arguments, we get 
\begin{equation*}
\mathrm{TV}(P_{r},P)=O( r^{\min \{ 1,\kappa _{s}+\kappa _{r}\} }) ~.
\end{equation*}

The convergence rate of the Hellinger distance follows from its relationship with the total variation distance, $\mathrm{H}^{2}(P_{r},P)\leq \mathrm{TV}(P_{r},P)$. \qed 

\end{document}